\documentclass[11pt]{article}

\usepackage[letterpaper,margin=1in]{geometry}
\usepackage[T1]{fontenc}
\usepackage[utf8]{inputenc}
\usepackage{lmodern}
\usepackage{microtype}
\usepackage{amsmath,amssymb,amsthm,mathtools}
\usepackage{algorithm}
\usepackage{algorithmic}
\usepackage{aliascnt}
\usepackage{thmtools,thm-restate}
\usepackage[hidelinks]{hyperref}
\usepackage[nameinlink,noabbrev]{cleveref}

\usepackage[suppress]{color-edits}
\addauthor[Yusuf]{yk}{blue}
\addauthor[Vasilis]{vl}{red}
\addauthor[Aditya]{ap}{cyan}
\addauthor[Shaddin]{sd}{magenta}

\crefname{section}{Section}{Sections}
\Crefname{section}{Section}{Sections}
\crefname{subsection}{Section}{Sections}
\Crefname{subsection}{Section}{Sections}
\crefname{theorem}{Theorem}{Theorems}
\Crefname{theorem}{Theorem}{Theorems}
\crefname{lemma}{Lemma}{Lemmas}
\Crefname{lemma}{Lemma}{Lemmas}
\crefname{proposition}{Proposition}{Propositions}
\Crefname{proposition}{Proposition}{Propositions}
\crefname{corollary}{Corollary}{Corollaries}
\Crefname{corollary}{Corollary}{Corollaries}
\crefname{observation}{Observation}{Observations}
\Crefname{observation}{Observation}{Observations}
\crefname{definition}{Definition}{Definitions}
\Crefname{definition}{Definition}{Definitions}
\crefname{example}{Example}{Examples}
\Crefname{example}{Example}{Examples}
\crefname{algorithm}{Algorithm}{Algorithms}
\Crefname{algorithm}{Algorithm}{Algorithms}

\theoremstyle{plain}
\newtheorem{theorem}{Theorem}[section]
\newtheorem{lemma}[theorem]{Lemma}
\newtheorem{proposition}[theorem]{Proposition}

\theoremstyle{definition}
\newtheorem{observation}[theorem]{Observation}
\newtheorem{definition}[theorem]{Definition}
\newtheorem{example}[theorem]{Example}

\newcommand{\BP}{\mathsf{BP}}
\newcommand{\RBP}{\mathsf{RBP}}
\renewcommand{\vec}[1]{\mathbf{#1}}

\newcommand{\OPT}{\mathsf{OPT}}
\newcommand{\alphaNR}{\alpha_{\mathrm{NR}}}
\newcommand{\alphaR}{\alpha_{\mathrm{R}}}
\newcommand{\floor}[1]{\left\lfloor #1 \right\rfloor}
\newcommand{\ceil}[1]{\left\lceil #1 \right\rceil}
\newcommand{\eps}{\varepsilon}

\title{Batched Pandora's Box}
\author{
    Shaddin Dughmi
    \thanks{
        University of Southern California,
        Email: {\tt shaddin@usc.edu}. Supported by the Air Force Office of Scientific Research under award number FA9550-24-1-0261. This work was done in part while the author was on sabbatical as the Carter and Tania Neild visiting professor at Northwestern University,
as well as a visiting professor in the Data Science Institute at the University of Chicago.
    }
    \and Yusuf Hakan Kalayci
    \thanks{
        University of Southern California,
        Email: {\tt kalayci@usc.edu}. Supported by the Air Force Office of Scientific Research under award number FA9550-24-1-0261.
    }
    \and Vasilis Livanos
    \thanks{
        University of Southern California,
        Email: {\tt vas.livanos@gmail.com}. Supported by the Air Force Office of Scientific Research under award number FA9550-24-1-0261. This work was done in part while the author was visiting Northwestern University.
    }
    \and Aditya Prasad
    \thanks{
        University of Chicago,
        Email: {\tt adityaprasad@uchicago.edu}.
    }
}
\date{}

\begin{document}

\maketitle

\begin{abstract}
Motivated by numerous parallelizable stochastic search problems, most notable and timely among them being LLM inference-time scaling, we propose and study batched versions of the Pandora's Box problem of Weitzman. In particular, boxes are opened in capacity-constrained batches, each batch has a setup cost, and all rewards in a batch are revealed together. We consider two different variants, motivated by different application environments: one where boxes are reusable (i.e., can provide multiple i.i.d.~samples) and another where they are not. For both variants we rule out most ``simple'' natural heuristics, and also formally prove NP-hardness of approximation in the traditional sense. We then relax the problem to allow bi-criteria approximations, with respect to both rewards and setup costs, where we exhibit constant approximation algorithms for both the reusable and non-reusable settings. This is obtained through a linear-programming relaxation of Pandora's Box problem, followed by randomized or Pipage rounding.
\end{abstract}

\section{Introduction}
\label{sec:introduction}

In several stochastic search problems, before we find a solution of sufficient
quality we need to pay a search cost. In marketing, a platform tests different
advertisements before committing to the best-performing one; in hyper-parameter
tuning, a practitioner trains their model with a menu of settings before selecting
a model; in inference-time reasoning, an inference system issues multiple LLM
queries or decoding strategies before choosing the best response. The classical
Pandora's Box problem, originally studied by Weitzman \cite{weitzman1979optimal}, captures
the sequential version of this tradeoff. Each option has an inspection cost and
an unknown reward, and the decision maker adaptively decides which option to
inspect next and when to stop.

In many of these applications, however, inspections are naturally parallelized. Ad variants in A-B testing are exposed
to users in cohorts, experiments are tested on a pool of machines, and queries to an LLM are submitted together to reduce latency. This parallelism creates a batched information structure, as the policy must commit
to a set of inspections before seeing any of their outcomes. A batch may also have
a setup cost, capturing the overhead of launching a new experimental framework, allocating
computing resources, or coordinating a marketing campaign. Thus batching trades off the cost of latency with the savings provided by adaptivity.

While batched feedback has been studied in bandit problems \cite{perchet2016batched,gao2019batched,grover2018delayed,verma2022delayedbo},
it has seen comparatively less study in the Pandora's Box problem. We initiate this
direction by studying a batched version of the Pandora's Box problem. In the \emph{Batched Pandora's Box model}, we are given $n$ boxes, each with a distribution of a hidden
reward $V_i$, as well as an opening cost $c_i$. A policy groups boxes to create
the next batch to be opened, and a feasible batch contains at most $k$ boxes for a
given batch \emph{capacity} $k$ and costs a fixed setup cost $T$ to open in addition to the individual opening costs. Once a batch is opened, all rewards in the batch are revealed simultaneously. One can then choose to stop and obtain the maximum reward observed over all opened batches or continue to open another batch.

We introduce two variants of this model: The \emph{non-reusable setting} where an
opened box is consumed and cannot be opened again, and the \emph{reusable setting} where each box can be opened multiple times and provides a fresh i.i.d.~sample each time. The distinction
matches the difference between one-time opportunities, such as allocating a
consumer cohort to an ad variant, and repeatable experiments, such as running
additional inference trials of the same LLM. Although this model is a
direct analogue of the classical problem with delayed feedback, the resulting
optimization problem differs significantly from classical Pandora's Box, requiring
us to examine which elements of Weitzman's solution to the classical problem
carry over to the batched model.

\subsection{Our Contributions}
\label{subsec:contributions}

We organize our results around three questions raised by batched feedback.

\paragraph{What survives of Weitzman's index rule under batched feedback?}
The answer depends on reusability. For non-reusable boxes, adaptive batch
formation is useful: a policy can choose later batches after seeing
earlier outcomes and thereby outperform every fixed menu of disjoint batches. This
advantage is bounded, however. We prove that every adaptive non-reusable policy
has a fixed-menu counterpart with at least half its value, using a random
root-to-leaf path of the adaptive decision tree in the spirit of stochastic
probing \cite{gupta2017adaptivity,bradac2019near}. For reusable boxes, we can
recover Weitzman's index structure. Each feasible batch $\vec{a}$ has a reservation
value $\sigma(\vec{a})$, and an optimal policy repeatedly opens a
maximum-reservation batch until its realized value exceeds that threshold.

\paragraph{Is exact or ordinary approximation tractable?}
The answer is negative, in a strong sense. Specifically, exact optimization is NP-hard
for the non-reusable problem for every fixed finite batch capacity $k\geq3$. Moreover, with
variable batch capacity, and even with $k=\infty$, both the non-reusable and reusable
models have a sign gap: it is NP-hard to distinguish $\OPT=0$ from
$\OPT>0$. The reductions come from \textsc{Subset Product}
\cite{garey1979computers}, where the input consists of integers $b_1,\ldots,b_n$ and a
target $Q$, and the question is whether some subset has product exactly $Q$.
For each \(b_i\), we create a box with a two-point distribution that gives value \(M\) with probability \(1-1/b_i\) and value \(0\) with probability \(1/b_i\). Hence a batch corresponding to a subset \(S\) fails with probability \(1/\prod_{i\in S} b_i\). By calibrating the setup and item costs, we make it so that a batch has positive surplus exactly when the product of the constituent boxes' $b_i$'s is the target \(Q\).
Therefore, finding any positive-value policy certifies a YES instance of
\textsc{Subset Product}. This sign gap rules out ordinary multiplicative
approximation under the original costs, since any positive approximation
guarantee would distinguish zero-optimum instances from positive-optimum
instances.

\paragraph{Can a discounted-cost relaxation recover useful algorithmic guarantees?}
The sign gap motivates a bi-criteria benchmark: compare against the original-cost
optimum, but evaluate the returned policy with discounted costs. For a policy
$\pi$, with maximum observed value $M_\pi$ and total cost $C_\pi$, let
\[
    \operatorname{val}_{\lambda}(\pi)
    =
    \mathbb{E}[M_\pi]-\lambda\,\mathbb{E}[C_\pi].
\]
An $(\alpha,\lambda)$ discounted-cost bi-criteria guarantee means that, for every
$\varepsilon>0$, there is an algorithm returning an \emph{efficient} policy $\widehat{\pi} = \widehat{\pi}(\eps)$ with
\[
    \operatorname{val}_{\lambda}(\widehat{\pi})
    \geq
    \alpha\operatorname{OPT}_1-\varepsilon,
\]
where $\operatorname{OPT}_1$ is the expected value of the optimal adaptive policy with the original (non-discounted) costs. For finite-support distributions, our algorithms for both variants are based on rounding a linear programming relaxation of the batched problem. In the non-reusable model, the LP records opening
probabilities, expected batch count, and probabilities that each box is the
maximum in a batch, and upper bounds
every nonempty adaptive policy. In a solution of this LP, the batch count can
be fractional, and so we need to round it first to an integral value. Then,
we use dependent rounding to partition the boxes into legal batches. We incur
a loss of $2(\sqrt{2}-1)$ from the batch count rounding and of $1 - 1/e$ from
the correlation gap of the dependent rounding scheme, yielding in total
\[
    \operatorname{val}_{\alphaNR}(\widehat{\pi})
    \geq
    \alphaNR\operatorname{OPT}_1-\varepsilon,
    \qquad \text{where } \quad 
    \alphaNR=(1-1/e)2(\sqrt{2}-1).
\]
In the reusable model, we reduce the problem to finding a single batch, corresponding to a multi-set of boxes, to sample repeatedly until another sample no longer pays for its cost. We design a feasibility LP
\(\mathrm{RLP}_\tau\) based on a candidate reservation value of the desired batch
that is feasible when there is a batch certifying a reservation value at least $\tau$.
Binary search finds the largest certified threshold up to a small additive error,
and independent slot rounding turns the certificate into an integral batch. In total,
we only incur a rounding loss of $\alphaR=1-1/e$, giving a stronger bound of
\[
    \operatorname{val}_{\alphaR}(\widehat{\pi})
    \geq
    \operatorname{OPT}_1-\varepsilon .
\]

\subsection{Related Work}
\label{subsec:related-work}
\paragraph{Pandora's Box and costly inspection.}
Weitzman's seminal work \cite{weitzman1979optimal} initiated study of the Pandora's Box setting and showed that independent boxes with inspection costs admit an optimal reservation-value policy.
Subsequent work extends this costly-search
view in several directions, including generalized objectives
\cite{olszewski2015general}, market-search coordination
\cite{kleinberg2016descending}, nonobligatory inspection
\cite{beyhaghi2019nonobligatory,fu2022nonobligatory,beyhaghi2022structure,scully2024local},
combinatorial or online feasibility constraints
\cite{singla2017price,boodaghians2020order,berger2023combinatorial,bowers2024nested,chawla2024combinatorial,chawla2025commitment},
and correlated or sample-access variants
\cite{chawla2019correlations,chawla2021approximating,gergatsouli2023correlations}. 
Another recent line studies online and learning variants of Pandora's Box under online-arrival, contextual, bandit, semi-bandit, and LLM-motivated feedback models
\cite{esfandiari2019online, gergatsouli2022online, atsidakou2024contextual, gatmiry2024bandit, agarwal2024semibandit, liu2025improved, kalayci2025optimal, belloni2026online}.
These works mainly change the distributional knowledge, feedback available across repeated rounds, contextual structure, or feasible selections. In contrast, our model is a known-distribution, single-instance Pandora's Box problem with independent rewards. We keep the best-observed-value objective and instead change the feedback structure: inspections must be launched in capacity-limited batches with setup costs.

\paragraph{Stochastic probing and rounding.}
Our structural and algorithmic tools also connect to stochastic probing and
submodular rounding. Stochastic probing studies the value of adaptivity when
random elements are probed subject to feasibility constraints, including
constant adaptivity gaps for submodular and XOS objectives and near-optimal
results for multi-value probing \cite{gupta2017adaptivity,bradac2019near}. Our
non-reusable adaptivity-gap proof uses the same random-path perspective, but
with Pandora-specific accounting through reservation values, capped rewards,
and fixed-menu Pandora policies. The LP algorithms use standard
correlation-gap and dependent-rounding ideas, including pipage rounding and
correlation-robust stochastic optimization
\cite{calinescu2011maximizing,agrawal2009correlation}.

\paragraph{Delayed feedback.}
Delayed feedback has been studied as a natural model in bandits and online
learning. One line of work studies how late observations affect regret in
stochastic, Bayesian, or
adversarial bandits, including general delay reductions, delayed conversions,
anonymous aggregated rewards, and unrestricted delays
\cite{guha2010delayed,joulani2013delayed,vernade2017delayed,pikeburke2017anonymous,thune2019unrestricted,wu2022thompson}.
Another line studies batched or parallel experimentation, where actions are chosen in groups before feedback from the group is available, for regret minimization, best-arm identification, or Bayesian optimization
\cite{perchet2016batched,gao2019batched,grover2018delayed,verma2022delayedbo}.

These works share the practical obstruction that decisions may have to be made
before earlier outcomes are observed, but their objective is typically
statistical: regret or sample complexity. Our model is a single-instance search problem.
Here delayed feedback enters through capacity-constrained batches with setup
costs, so the delay changes the economics and commitment structure of costly
inspection rather than only the rate at which a learner receives information.

\section{Preliminaries}
\label{sec:preliminaries}

In the classical Pandora's Box problem, there are $n$ boxes numbered
$[n]=\{1,\ldots,n\}$. Box $i$ has an independent non-negative reward $V_i$ with
finite mean and a deterministic opening cost $c_i\geq0$. Opening box $i$ costs
$c_i$ and immediately reveals $V_i$. Weitzman's rule assigns each box a
reservation value
\[
\tau_i=\sup\left\{\tau\in\mathbb{R}:
  \mathbb{E}\bigl[(V_i-\tau)^+\bigr]\geq c_i
\right\},
\]
opens boxes in decreasing order of $\tau_i$, and stops when the best observed
reward is at least as large as the largest remaining reservation value
\cite{weitzman1979optimal}.

\paragraph{Model.} We study a batched version of this problem. We group the boxes into batches
where each batch contains at most $k$ boxes and, to open a nonempty batch,
we have to pay a setup cost $T \geq 0$ in addition to the individual opening
costs. When a batch is opened, all rewards in it are revealed together. We
consider two variants: a non-reusable model, as in the classical problem,
where opened boxes are consumed and thus batches partition the boxes, and
a reusable model, where each box can provide multiple i.i.d.~samples and
thus batches correspond to multisets of boxes.

\begin{definition}[Non-Reusable Batched Pandora's Box, $\BP(T,k)$]
A feasible batch is a nonempty subset $B$ of currently unopened boxes with
$|B|\leq k$. Opening $B$ costs
\[
K(B)=T+\sum_{i\in B}c_i,
\]
reveals $(V_i)_{i\in B}$, and consumes all boxes in $B$. If $k=\infty$, the
cardinality constraint is omitted.
\end{definition}

Since the batches partition the $n$ boxes in the non-reusable model, $\BP(T,k)$
is identical to $\BP(T,n)$ whenever $k\geq n$.

\begin{definition}[Reusable Batched Pandora's Box, $\RBP(T,k)$]
Each box type \(i\in[n]\) has a strictly positive sampling cost \(c_i>0\).\footnote{This assumption is relevant only when \(k=\infty\); it rules out degenerate instances in which arbitrarily many copies of a zero-cost type can be included without increasing the batch cost.}
A feasible reusable batch is a vector
$\vec{a}\in\mathbb{Z}_{\geq0}^n$ of multiplicities with
$1\leq \|\vec{a}\|_1\leq k$. Opening $\vec{a}$ costs
\[
K(\vec{a})=T+\sum_i a_i c_i
\]
and draws $a_i$ i.i.d.~samples of type $i$. The reward observed in batch $\vec{a}$
is
\[
X_{\vec{a}}=\max_{i\in[n]}\max_{1\leq r\leq a_i}V_{i,r},
\]
where $V_{i,r}$ denotes the $r$-th sample of type $i$. Since boxes (and thus batches)
are reusable, they remain available to be used in later rounds. If $k=\infty$,
the cardinality constraint is omitted.
\end{definition}

For a policy $\pi$ in any of these models, let $M_\pi$ be the maximum reward
observed by the policy, with $M_\pi=0$ if no box is opened, and let $C_\pi$ be
the total cost paid, including setup and individual opening costs. As in the
classical problem, we evaluate the performance of a policy $\pi$ by its expected
maximum reward minus its expected total cost. We define the discounted-cost value
\[
    \operatorname{val}_\lambda(\pi)
    =
    \mathbb{E}[M_\pi]-\lambda\,\mathbb{E}[C_\pi],
\]
where $\lambda\in[0,1]$. Thus, $\lambda=1$ corresponds to the original-cost
objective, while $\lambda<1$ corresponds to evaluating the policy after discounting
the costs. We refer to $\operatorname{val}_1(\pi)$ as the \emph{utility} of policy \(\pi\). Similarly, the \emph{utility} of a batch --- also called its \emph{one-shot surplus} --- is its expected maximum reward minus its opening cost. We write $\operatorname{OPT}_1$ for the optimal original-cost value,
taking the empty policy when the optimum is zero.
The source of the main departure of our batched model from the classical Pandora's Box is the setup cost.
\begin{observation}\label{prop:zero-setup}
When $T=0$, any chosen batch can be simulated by opening its boxes sequentially,
ignoring intermediate observations until the batch is complete; the same
applies to the fresh samples in a reusable batch. Since singleton batches (batches of size $1$) are
feasible, for every $k\geq1$, $\BP(0,k)$ has the classical singleton value (i.e. $\BP(0,k) = \BP(0, 1)$), and
if $c_i>0$ for every $i$, then $\RBP(0,k)$ has the reusable singleton value (i.e. $\RBP(0, k) = \RBP(0, 1)$).
\end{observation}

\paragraph{Computational considerations.}
We represent each distribution by a finite list of value-probability pairs
$(v_{i\ell},p_{i\ell})$, where $v_{i\ell}\geq0$, $p_{i\ell}>0$, and
$\sum_\ell p_{i\ell}=1$. We use $L_i$ for the (non-zero) support of the distribution
of box $i$, i.e. $\ell \in L_i \implies v_{i\ell} > 0$. The hardness reductions
use two-point distributions, while all approximation statements are stated for this
finite-support input model.

\paragraph{Submodular functions.}
For a set function $f:2^{[n]}\to\mathbb{R}_{\geq0}$, we say that $f$ is
normalized if $f(\varnothing)=0$, monotone if $f(A)\leq f(B)$ for all
$A\subseteq B$, and submodular if, for all $A\subseteq B\subseteq[n]$ and
$i\notin B$,
\[
f(A\cup\{i\})-f(A)\geq f(B\cup\{i\})-f(B).
\]
The function $F(x)=\mathbb{E}[f(S_x)]$ is called the multilinear extension of $f$,
where $S_x$ includes each element $i$ independently with probability $x_i$.
\section{Structural Observations}
\label{sec:structural}

We begin by examining how the index-based structure of Weitzman's policy changes under batching. 
A natural first attempt greedily merges boxes in order of decreasing reservation values. 
We show that this type of sort-and-merge consecutive boxes policy fails badly: in our construction, the only positive utility batch combines a non-consecutive group of boxes. Here, the optimal policy's batches can be made arbitrarily far apart. This rules out both greedy and dynamic programming approaches for partitioning the boxes into consecutive batches in the reservation-value order.  

Next, we examine the strength of adaptivity. In the non-reusable model, we present an example that establishes a gap between the value attained by adaptive policies and fixed-menu policies that commit to a partition of boxes into batches and run Weitzman's policy on the fixed batches. Then, we show that this adaptivity gap between the two types of policies can be at most a factor of $2$. Finally, we show that in the reusable setting, there is no gap between the optimal adaptive and non-adaptive policy. As a nice byproduct, we recover an index principle in the style of Weitzman: specifically, the optimal policy repeats the batch with the highest reservation value until the batch value exceeds the reservation threshold.

\subsection{Failure of Simple Batching Rules}

A natural first attempt is to import Weitzman's ordering rule into the batched
problem. Compute each box's reservation value as if it were opened alone, sort
boxes in decreasing order of these values, and then form batches from
consecutive blocks of this order. This preserves the intuition that high-index
boxes should be inspected early and, in the batched setting, possibly inspected
together. The next example shows that this rule misses a batched
effect: setup costs can make two separated boxes valuable together even though
every consecutive batch in the reservation-value order is unprofitable.

Take \(T=10\) and \(k=2\). Let \(A\) have item cost \(c_A=1\), with
\(\Pr[V_A=100]=0.1\) and \(\Pr[V_A=0]=0.9\). Let \(B\) be deterministic with
value \(9\), with item cost \(c_B=1\). Individually, neither box pays for the
setup cost. However, together they have expected maximum
\[
    \mathbb{E}[\max\{V_A,V_B\}]
    =
    0.1\cdot 100+0.9\cdot 9=18.1,
\]
so the batch \(\{A,B\}\) has positive utility.
We will design the instance so that the target boxes A and B will be arbitrarily far apart under Weitzman's reservation-value ordering.

Now consider a deterministic decoy box \(D\). Let \(D\) have value \(100\)
and item cost \(c_D=91\). Its reservation value is worse than \(A\)'s but better
than \(B\)'s, so the sorted order is $A\succ D\succ B$. At the same time, \(D\) never contributes a useful net value: it duplicates the
high outcome of \(A\), but its item cost is almost as large as its value. One
can verify that every feasible consecutive batch has negative utility, while \(\{A,B\}\) has positive utility. Therefore, no consecutive
batching algorithm, whether greedy or based on dynamic programming over
intervals, can obtain positive utility on this example. The exact calculations are given in
\cref{subsec:consecutive-batching-lower-bound}, which also proves a stronger
bi-criteria lower bound: for every batch capacity \(k\ge 2\) and every
\(\alpha>1/2\), no algorithm restricted to consecutive batches in the
reservation-value order can guarantee an
\((\alpha,\alpha)\)-bi-criteria approximation.

\subsection{Non-Reusable Model: Adaptivity Gap}

A fixed-menu policy first commits to a disjoint family $\mathcal{F}$ of feasible
batches. It then treats each batch $B\in\mathcal{F}$ as one box in a finite
classical Pandora instance, with reward $X_B=\max_{i\in B}V_i$ and cost
$K(B)$, and applies Weitzman's rule. Let
$\operatorname{OPT}_{\mathrm{fix}}(T,k)$ denote the best value of any such
fixed-menu policy. This benchmark preserves the Pandora stopping rule while
removing the branching structure of adaptive batch choice.

The natural question is whether the best fixed menu can match the value of the
best adaptive policy. For $k = 1$ the answer is yes, by the classical Pandora
rule. For larger batches, fixed menus lose the ability to branch on an early
observation before deciding which later batch to open. The following example
shows that this restriction can reduce value.

\begin{example}[Adaptive branching can help]\label{ex:strict-gap}
Consider three boxes with setup cost \(T=1/5\) and batch capacity \(k=2\). Box \(1\)
is deterministic with value \(1/2\). Box \(2\) is \(0\) with probability \(2/5\), \(1/2\) with probability \(2/5\), and \(9\) with probability \(1/5\). Box \(3\) is \(10\) with probability \(1/5\), and \(0\) otherwise. The item costs are \(c_1=c_2=\frac15\) and \(c_3=1\).

The adaptive policy first opens box \(2\). If box \(2\) realizes as \(9\), the
policy stops. If box \(2\) realizes as \(1/2\), it opens box \(3\) alone. If box
\(2\) realizes as \(0\), it opens the batch \(\{1,3\}\). The point is that box
\(1\) is useful only if box \(3\) realizes as $0$, i.e. when we don't already have a guarantee from box $2$.

A fixed menu cannot make this branch-dependent choice. It must decide in
advance whether box \(1\) is grouped with box \(3\), separated from box \(3\), or
grouped with box \(2\). For this instance, the adaptive policy has value $57/25 = 2.28$, whereas the best fixed-menu policy has value $279/125 = 2.232$. 
Thus adaptive branching strictly improves over every fixed menu. Appendix~\ref{app:strict-gap-example} gives a tuned version of the same
signal/fallback/risky-box template, with all parameters rounded to two decimal places; in that instance, the adaptive policy attains about \(1.193\) times the value of the best fixed-menu policy.
\end{example}

Nevertheless, the value of adaptive batch formation is bounded: every adaptive
policy is within a constant factor of the best fixed-menu policy.

\begin{theorem}[Non-Reusable Adaptivity Gap]\label{thm:bp-gap}
For every $T$ and $k$,
\[
\operatorname{OPT}_1^{\BP(T,k)}\leq
2\cdot\operatorname{OPT}_{\mathrm{fix}}(T,k).
\]
\end{theorem}

The proof is existential and is given in \cref{app:structural-proofs}. It
utilizes a random-path rounding argument in the style of stochastic probing
\cite{gupta2017adaptivity,bradac2019near}: one samples a root-to-leaf path of
the adaptive decision tree on an independent ``ghost'' instance and then plays
that fixed path on the real instance. For this style of argument, we must ensure that along any realization of box values, the batches consist of disjoint sets of boxes. 
Then, we use Weitzman's capped-value analysis to turn this path into a fixed-menu policy losing at most a factor two.

\subsection{Reusable boxes: one repeated batch}

We turn our attention to the reusable setting. Since opening a reusable
batch does not consume future access to it, each reusable batch can be viewed as
a box with infinitely many independent copies. The dynamic problem therefore
reduces to identifying the batch with the largest reservation value. For a reusable batch $\vec{a}$, recall that $X_{\vec{a}}$ is the best reward among the fresh
copies drawn by opening $\vec{a}$, and $K(\vec{a})$ is its setup-plus-sampling cost. If
$0<K(\vec{a})\leq\mathbb{E}[X_{\vec{a}}]$, define its reservation value $\sigma(\vec{a})$ by $K(\vec{a})=\mathbb{E}\bigl[(X_{\vec{a}}-\sigma(\vec{a}))^+\bigr]$.
Batches with $K(\vec{a})>\mathbb{E}[X_{\vec{a}}]$ are overpriced. Let
\[
\sigma^\star=
\max\Bigl(\{0\}\cup
\{\sigma(\vec{a}):\vec{a}\text{ is feasible and not overpriced}\}\Bigr).
\]

\begin{theorem}[Reusable Reservation Index]\label{thm:reusable-index}
Assume values lie in $[0,V_{\max}]$ and every non-empty feasible batch has
positive opening cost. For $k=\infty$, also assume every $c_i>0$. If
$\sigma^\star=0$, opening nothing is optimal. If $\sigma^\star>0$, then for any feasible
batch $\vec{a}^{\star}$ with $\sigma(\vec{a}^{\star}) = \sigma^*$, the policy that repeatedly opens $\vec{a}^{\star}$ until it observes a reward exceeding $\sigma^\star$ is optimal.
\end{theorem}

The theorem is the reusable analogue of Weitzman's index rule. A copy of the
best reservation batch remains available after every failed attempt, so the
index policy continues to choose it. The boundedness and positive-cost assumptions
ensure that the relevant batches that are not overpriced are finite in number; the
formal details are in \cref{app:structural-proofs}.

The repeated-batch value equals the reservation value. If batch $\vec{a}$ has
reservation value $\sigma(\vec{a})$ and
$p=\Pr[X_{\vec{a}}>\sigma(\vec{a})]$, then repeatedly opening it until the first
success gives
\[
\begin{aligned}
\mathbb{E}[X_{\vec{a}}\mid X_{\vec{a}}>\sigma(\vec{a})]
-\frac{K(\vec{a})}{p}
&=
\frac{\mathbb{E}\!\left[X_{\vec{a}}\mathbf{1}\{X_{\vec{a}}>\sigma(\vec{a})\}\right]}{p}
-\frac{K(\vec{a})}{p} =
\sigma(\vec{a})+
\frac{\mathbb{E}[(X_{\vec{a}}-\sigma(\vec{a}))^+]-K(\vec{a})}{p} \\
&= \sigma(\vec{a}).
\end{aligned}
\]
The last equality uses
\(\mathbb{E}[(X_{\vec{a}}-\sigma(\vec{a}))^+]=K(\vec{a})\), by the definition of
\(\sigma(\vec{a})\).
\section{Hardness Results}
\label{sec:hardness}

In this section, we show that the various settings of Batched Pandora's Box are NP-hard. Broadly, our reductions construct instances of Batched Pandora's Box in which even building a batch with positive utility encodes NP-hard search problems.

\paragraph{Product gadget.}
Our reductions use a common product gadget.
Each box has the same positive prize $M$, and the box corresponding to \emph{label} $b_i$ fails with probability $1/b_i$.
A batch whose labels multiply to weight $w = \prod_{i \in B} b_i$ then fails with probability $1/w$.
We set the item costs and the batching cost carefully as a function of the weights of the boxes, ensuring a batch offers positive opening utility if and only if its labels multiply to a prescribed target value $Q$.
Consequently, finding a batch with positive opening utility encodes an NP-hard search problem, yielding hardness for Batched Pandora's Box.
For clarity, we initially present the proof using exact logarithms, and then show that rational approximations to the logarithms (of polynomial length) still preserve this property.
The formal proofs are presented in Section \ref{sec:hardness-proofs}.


\paragraph{Non-reusable boxes.}
We first investigate the non-reusable problem. Exact optimization is already
intractable for every fixed finite batch capacity at least three.

\begin{theorem}[Fixed Finite-Batch-Capacity Exact Hardness]
\label{thm:bp-fixed-hard}
For every fixed $k\geq3$, exact optimization of $\BP(T,k)$ is NP-hard, even when
$T>0$ and all boxes have support $\{0, M\}$ for a shared $M$. Equivalently,
deciding whether $\operatorname{OPT}_1\geq\Theta$ for a rational threshold
$\Theta$ is NP-hard.
\end{theorem}

We reduce from the strongly NP-hard Numerical 3-Dimensional Matching problem. 
The instance consists of three classes of items $X, Y, Z$ and a size function $s(e), e \in X \cup Y \cup Z$, and the goal is to partition the items into triples $(x, y, z)$ satisfying $s(x) + s(y) + s(z) = B$. 

For batch size $k=3$, we encode each item $e$ as a box whose label contains the factor $2^{s(e)}$. 
To encode the class constraint, we assign a distinct identifying prime to each of $X,Y,Z$, and include the corresponding prime in the label of every item from that class.
We then set the target product to be the product of the three class primes times $2^B$. By unique factorization, a batch has the target product if and only if it contains one item from each class whose sizes sum to $B$. For $k>3$, we add dummy zero-size classes to fill the remaining positions in each batch.

The product gadget makes precisely the target tuples profitable. Since all target tuples
have the same success probability and the same cost, the optimal value is a
strictly increasing function of the maximum number of disjoint target tuples,
and exact optimization distinguishes whether a perfect numerical matching
exists.

The next question is whether one can at least obtain an ordinary constant-factor
approximation. When batch capacity is part of the input, the product gadget gives a
stronger obstruction: distinguishing zero value from positive value is already
NP-hard.

\begin{theorem}[Non-Reusable Sign-Gap Decision Hardness]
\label{thm:bp-sign-hard}
For $\BP(T,k)$ with $k$ part of the input, and for $\BP(T,\infty)$, the promise
problem of distinguishing
\[
\operatorname{OPT}_1=0
\qquad\text{from}\qquad
\operatorname{OPT}_1>0
\]
is NP-hard, even on common-prize two-point instances with $T>0$.
\end{theorem}

Here the reduction is from \textsc{Subset Product}: given labels
$a_1,\ldots,a_n$ and target $Q$, decide whether some subset has product $Q$.
The batch capacity is set to the number of source items. If a target subset exists,
opening the corresponding batch once gives positive value; otherwise every
nonempty batch has negative one-shot surplus. Because all boxes have a common
prize, adaptive stopping cannot turn these negative-surplus attempts into
positive value: before the first success the policy observes only zeros, and
the expected reward from any sequence of attempts is dominated by its expected
cost.

\paragraph{Reusable boxes.}
For reusable boxes, the same product idea applies, but we must prevent a target batch from using multiple copies of the same source item. 
Broadly, for each source item $i$, we introduce a private prime $p_i$ and create two boxes whose labels correspond to $p_i$ and $p_i a_i$.
We set the target product to $Q\prod_i p_i$. 
Since $p_i$ appears in no other label and has exponent one in the target, any multi-set whose labels multiply to the target must contain exactly one of $p_i$ and $p_i a_i$ for each $i$.
Choosing $p_i$ corresponds to excluding $a_i$ from the subset, while choosing $p_i a_i$ corresponds to including it. 
After canceling the private primes, the remaining condition enforces that the product of the selected $a_i$'s equals $Q$. 
Finally, the product gadget then gives the reusable sign gap.

\begin{theorem}[Reusable Sign-Gap Decision Hardness]
\label{thm:rbp-hard}
For $\RBP(T,k)$ with finite $k$ as part of the input, and for $\RBP(T,\infty)$,
the promise problem of distinguishing
\[
\operatorname{OPT}_1=0
\qquad\text{from}\qquad
\operatorname{OPT}_1>0
\]
is NP-hard, even on common-prize two-point instances with $T>0$.
\end{theorem}

For unbounded batch capacity, item costs make all sufficiently large multisets
automatically overpriced, so the reduction only has to protect a finite range of
multiplicities by the product separation.

\section{Bi-criteria Approximation Results}
\label{sec:approximation}

The hardness results of the previous section preclude standard positive-factor
approximation guarantees under the original costs. We therefore turn to discounted-cost
bicriteria guarantees. Recall that
\(\operatorname{val}_\lambda(\pi) = \mathbb{E}[M_\pi]-\lambda\,\mathbb{E}[C_\pi]\), and
that an algorithm has an $(\alpha,\lambda)$ guarantee if, for every requested $\varepsilon>0$,
it returns a feasible policy $\widehat{\pi} = \widehat{\pi}(\eps)$ with $\operatorname{val}_\lambda(\widehat{\pi})
  \geq
\alpha\operatorname{OPT}_1-\varepsilon$. Thus, the algorithm is compared with the original-cost optimum, while its returned policy is evaluated after discounting costs.

\subsection{Non-Reusable Boxes}
\label{subsec:approx-nonreusable}

Consider first the non-reusable model. Although an adaptive policy may have a
possibly exponential decision tree, before any values are realized, it induces
only ex-ante statistics: the probability $x_i$ that each box is ever opened and
the expected number $z$ of batches it opens. These statistics determine the
expected cost. The relaxation therefore only needs to upper-bound the final
reward that can be attributed to the opened boxes. For each box $i$, we omit the zero atom and, for $\ell \in L_i$, write
\[
\Pr[V_i=v_{i\ell}]=p_{i\ell},
\qquad \text{where} \qquad 
v_{i\ell}>0.
\]
When box $i$ is the final selected value and realizes as the point $v_{i\ell}$
in its support, we assign the probability mass of this event to a variable
$y_{i\ell}$. This ``winner'' event can occur only if box $i$ is opened and
realizes as $v_{i\ell}$, and at most one positive value is selected overall. These two observations give the following relaxation:
\begin{equation*}
\tag{$\mathrm{LP}$}
\begin{aligned}
\max\quad
&\sum_i\sum_{\ell\in L_i} v_{i\ell}y_{i\ell}
  -Tz-\sum_i c_i x_i \\
\text{s.t.}\quad
&1\leq z\leq n,\\
&\sum_i x_i\leq kz,\\
&0\leq y_{i\ell}\leq p_{i\ell}x_i &&\forall i,\ell\in L_i,\\
&0\leq x_i\leq1 &&\forall i,\\
&\sum_i\sum_{\ell\in L_i}y_{i\ell}\leq1
\end{aligned}
\end{equation*}
The constraint $\sum_i x_i\leq kz$ is the ex-ante version of batch capacity:
if a policy opens $z$ batches in expectation, it cannot open more than $kz$
boxes in expectation.

At first glance, one could think that we could write $n$ different versions
of this LP, one per candidate value of $z$, since any policy always opens an
integer number of batches, and take the maximum of these $n$ LPs. This, in
principle, would save us the extra loss we incur from rounding $z$. Unfortunately,
such an approach does not constitute a valid upper bound to the optimal adaptive
policy, since the latter could open a different number of batches depending on
the history of observed values and thus, in expectation, a fractional number of
batches overall. For this reason, allowing $z$ to be fractional is crucial so
that the LP can capture the value of the optimal adaptive algorithm. Also, the
LP is a relaxation for the optimal adaptive algorithm, because the atom masses are
allowed to be arranged fractionally into a final winning value. After fixing the
opening marginals $x$, this upper bound on the expected reward has the equivalent threshold form
\[
B(x)=
\int_0^\infty
\min\left\{1,\sum_i x_i\Pr[V_i>t]\right\}\,dt .
\]
Sampling which boxes to open independently with probabilities given by the $x_i$'s
guarantees a $(1-1/e)$-fraction of this envelope at every threshold. Afterwards, pipage
rounding algorithm converts the opening marginals $x$ into an actual set of opened boxes,
while respecting the batch budget certified by $z$, without decreasing the expected
monotone submodular reward.

We also round $z$ to one of its two neighboring integers and scale the opening
marginals accordingly. Once the number of batches is integral, a selected set of at
most $k\bar z$ boxes can be partitioned arbitrarily into $\bar z$ legal batches.
Consequently, the rounding incurs only the batch-count factor $2(\sqrt2-1)$ from
rounding the batch count and no additional contention loss inside individual batches.

\begin{algorithm}[t]
\caption{Non-reusable LP rounding}
\label{alg:lp-rounding}
\begin{algorithmic}
\REQUIRE Opening marginals $x$, fractional batch count $z\geq1$, and capacity $k$
\STATE Set $\theta=z-\floor{z}$
\STATE Draw $\bar z=\floor{z}$ with probability $1-\theta$, and
       $\bar z=\ceil{z}$ with probability $\theta$
\STATE Set $q_i=\min\{1,(\bar z/z)x_i\}$ for every box $i$
\STATE Apply pipage rounding over the rank-$(k\bar z)$ uniform matroid to obtain
       a set $S$ with $|S|\leq k\bar z$
\STATE Partition $S$ arbitrarily into at most $\bar z$ groups of size at most $k$ each, open all
       nonempty groups, and stop
\end{algorithmic}
\end{algorithm}

\begin{restatable}[Finite-Support Non-Reusable Guarantee]{theorem}{generalnoreusablethm}
\label{thm:general-noreusable}
For every finite-support non-reusable instance and every \(\varepsilon>0\),
there is a randomized polynomial-time algorithm that returns a policy
\(\widehat{\pi}=\widehat{\pi}(\varepsilon)\) satisfying
\[
    \operatorname{val}_{\alphaNR}(\widehat{\pi})
    \ge
    \alphaNR\,\operatorname{OPT}_1-\varepsilon,
    \qquad
    \alphaNR
    =
    \left(1-\frac1e\right)2(\sqrt2-1).
\]
The algorithm solves \(\mathrm{(LP)}\), returns the empty policy if its optimum is
non-positive, and otherwise applies \cref{alg:lp-rounding} to an optimal
solution.
\end{restatable}

The proof follows this outline. The details are in \cref{subsec:proof-nonreusable-approx}.

\subsection{Reusable Boxes}
\label{subsec:approx-reusable}

The reusable case starts from a different structural fact. By
\cref{thm:reusable-index}, an optimal original-cost policy repeats one
maximum-reservation batch. Thus, it suffices to find one certifying
multiset of at most $k$ fresh i.i.d.~samples. The case $T=0$ is
handled exactly by \cref{prop:zero-setup}, so assume $T>0$ below
and let $V_{\max}$ be the largest support value.

For a threshold $\tau$, write $f_\tau(\vec{a})=\mathbb{E}[(X_{\vec{a}}-\tau)^+]$. If one batch pays for its discounted cost using the expected excess above
$\tau$, then repeatedly opening it guarantees discounted value at least
$\tau$.

\begin{lemma}[Threshold Certificate]
\label{lem:reusable-threshold-certificate}
Fix $\lambda\in(0,1]$. If a nonempty reusable batch $\vec{a}$ satisfies
\[
f_\tau(\vec{a})\geq \lambda K(\vec{a}),
\]
then repeatedly opening $\vec{a}$ until the first realization
$X_{\vec{a}}>\tau$ has $\lambda$-discounted value at least $\tau$.
\end{lemma}

This observation reduces the reusable algorithm to a threshold search. For a
candidate $\tau$, the task is to determine whether some batch can certify it,
i.e. whether there is a batch with reservation value at least $\tau$.
The fractional relaxation is similar to the non-reusable LP, except that it
is turned into a feasibility LP and $x_i$ is now a fractional number of copies
of type $i$ in a single reusable batch: 
\begin{equation*}
\tag{$\mathrm{RLP}_\tau$}
\begin{aligned}
\text{find}\quad&x, y\\
\text{s.t.}\quad
&\sum_i x_i\leq k,\\
&0\leq y_{i\ell}\leq p_{i\ell}x_i &&\forall i,\ell\in L_i,\\
&0\leq x_i &&\forall i,\\
&\sum_i\sum_{\ell\in L_i}y_{i\ell}\leq 1 \\
&\sum_i\sum_{\ell\in L_i}(v_{i\ell}-\tau)^+y_{i\ell}
  \geq T+\sum_i c_i x_i
\end{aligned}
\end{equation*}

We say that \(\tau\) is LP-certified if \(\mathrm{RLP}_\tau\) is feasible. The \(y\)-variables form a fractional one-winner certificate: the constraints \(y_{i\ell}\le p_{i\ell}x_i\) encode the marginal availability of atom \(v_{i\ell}\), while \(\sum_{i,\ell}y_{i\ell}\le 1\) encodes that one batch contributes only one maximum value. Thus feasibility of \(\mathrm{RLP}_\tau\) certifies that the fractional excess above \(\tau\) can pay for the batch cost. Lowering \(\tau\) only increases the coefficients \((v_{i\ell}-\tau)^+\) and makes the certificate constraint easier to satisfy. Therefore, certification is monotone and binary search can find the largest certified threshold up to an additive error.

When a threshold is certified fractionally, we round by filling the $k$ batch
slots independently according to the copy counts $x_i$. At level $t\geq\tau$,
one slot exceeds $t$ with probability
$\sum_i x_i\Pr[V_i>t]/k$, so the rounded batch achieves an
$\alphaR=1-1/e$ fraction of the fractional excess envelope. The method of conditional
expectations derandomizes the slot choices without losing the certificate.

\begin{algorithm}[t]
\caption{Reusable threshold search}
\label{alg:reusable-threshold}
\begin{algorithmic}
\REQUIRE A finite-support reusable $\RBP(T,k)$ instance and accuracy $\varepsilon>0$
\STATE Set $\ell=0$, $u=V_{\max}$, $\widehat{\tau}=0$, and $\widehat{\vec{a}}=\vec{0}$
\FOR{$\ceil{\log_2(V_{\max}/\varepsilon)}$ iterations}
  \STATE Set $\tau=(\ell+u)/2$ and solve \(\mathrm{RLP}_\tau\)
  \IF{$\tau$ is LP-certified}
    \STATE Round feasible LP solution into a certifying
           batch $\vec{a}_\tau$ using the method of conditional expectation 
    \STATE Set $\ell=\tau$, $\widehat{\tau}=\tau$, and
           $\widehat{\vec{a}}=\vec{a}_\tau$
  \ELSE
    \STATE Set $u=\tau$
  \ENDIF
\ENDFOR
\IF{$\widehat{\vec{a}}=\vec{0}$}
  \RETURN the empty policy
\ENDIF
\STATE Repeatedly open $\widehat{\vec{a}}$ until its value exceeds $\widehat{\tau}$
\end{algorithmic}
\end{algorithm}

\begin{restatable}[Finite-Support Reusable Guarantee]{theorem}{reusablegeneralthm}
\label{thm:reusable-general}
For every finite-support reusable instance with finite batch capacity \(k\) and
every \(\varepsilon>0\), \cref{alg:reusable-threshold} runs in polynomial time
and returns a policy \(\widehat{\pi}=\widehat{\pi}(\varepsilon)\) satisfying
\[
    \operatorname{val}_{\alphaR}(\widehat{\pi})
    \ge
    \operatorname{OPT}_1-\varepsilon,
    \qquad
    \alphaR=1-\frac1e.
\]
Equivalently, the algorithm provides a
\((1,\alphaR)\) discounted-cost bicriteria guarantee.
\end{restatable}

The comparison with the optimum follows from the reusable reservation-index
theorem. If $\vec{a}^\star$ is an optimal original-cost batch with reservation
value $\operatorname{OPT}_1$, then it certifies every threshold
$\tau\leq\operatorname{OPT}_1$ in the feasibility LP. Binary search therefore reaches a
threshold within $\varepsilon$ of the optimum, and the rounded certifying batch
is converted into a repeated policy by
\cref{lem:reusable-threshold-certificate}. The formal proof appears in
\cref{subsec:proof-reusable-approx}.

\section{Proofs of Hardness Results}
\label{sec:hardness-proofs}

We will repeatedly use two numerical primitives. First, define
\[
D(z)=\ln z+\frac{1}{z}-1.
\]
This function is minimized at $z=1$. Moreover, when $w$ and $Q$ are positive
integers with $w\ne Q$, the value $D(w/Q)$ is bounded away from zero by an
inverse polynomial in $Q$.

\begin{lemma}[Analytic Product Separation]
\label{lem:analytic-product-separation}
For every integer $Q\geq8$ and every positive integer $w\ne Q$,
\[
D(w/Q)\geq \frac{1}{16\cdot Q^2}.
\]
\end{lemma}

\begin{proof}
We have
\[
D(1)=D'(1)=0,
\qquad
D'(z)=\frac{z-1}{z^2},
\qquad
D''(z)=\frac{2-z}{z^3}.
\]
Thus $D$ decreases on $(0,1)$ and increases on $(1,\infty)$. On
$[1/2,3/2]$, $D''(z)\geq4/27$, so Taylor's theorem around $1$ gives
\[
D(z)\geq \frac{2}{27}\cdot(z-1)^2.
\]
If $w/Q\in[1/2,3/2]$, then $w\ne Q$ and integrality imply
$|w/Q-1|\geq1/Q$, hence
\[
D(w/Q)\geq \frac{2}{27\cdot Q^2}
>
\frac{1}{16\cdot Q^2}.
\]
Outside this interval, monotonicity gives
\[
D(w/Q)\geq \min\{D(1/2),D(3/2)\}
>
\frac{1}{100}
\geq \frac{1}{16\cdot Q^2},
\]
where the last inequality uses $Q\geq8$.
\end{proof}

Second, we will use rational approximations to logarithms so that all
constructed costs are rational and polynomially representable.

\begin{lemma}[Rational Log Approximation]
\label{lem:rational-log-surrogates}
Let $m_1,\ldots,m_r$ be positive integers and let $\eta>0$ be rational. In time
polynomial in the input bit length and $\log(1/\eta)$, one can compute rational
numbers $\ell_1,\ldots,\ell_r$ such that
\[
|\ell_i-\ln m_i|\leq \eta
\qquad\text{for all }i.
\]
\end{lemma}

\begin{proof}
The approximation claim is the standard bit-complexity computation of elementary
functions: after range reduction, logarithms can be approximated to $b$ bits by
classical polynomial-time series or binary-splitting algorithms.
\end{proof}

\subsection{Proof of Theorem~\ref{thm:bp-fixed-hard}: Fixed finite-capacity non-reusable boxes}

\begin{proof}[Proof of Theorem~\ref{thm:bp-fixed-hard}]
We first describe the reduction using exact logarithmic costs and ignore
representation-size issues; the final paragraph replaces the logarithms by
rational approximations and checks that the constructed instance has polynomial
encoding length. We reduce from
strongly NP-complete Numerical 3-Dimensional Matching
\cite{garey1979computers}. An instance consists of three classes $X,Y,Z$, each
of size $m$, positive integer sizes $s(e)$, and a target $B$. The question is
whether the elements can be partitioned into $m$ triples $(x,y,z)$ satisfying
\[
s(x)+s(y)+s(z)=B.
\]
Since Numerical 3DM is strongly NP-complete, we may restrict to instances in
which $B$ and all sizes $s(e)$ are bounded by a polynomial in $m$. We also
assume $m\geq2$, since the case $m=1$ is decidable directly.

Fix $k\geq3$. For $k=3$, use the three original classes. For $k>3$, add
$k-3$ dummy classes, each containing $m$ dummy elements of size zero. Choose
$k$ distinct odd primes $p_1,\ldots,p_k$, one per class, and use $2$ as a
separate encoding prime. For an element $e$ in class $j$, define
\[
b_e=p_j\cdot 2^{s(e)}.
\]
For dummy elements, $s(e)=0$. Set
\[
Q=\left(\prod_{j=1}^k p_j\right)\cdot 2^B.
\]

Call a batch \emph{target} if its label product is $Q$. By unique factorization,
target batches are exactly the feasible $k$-tuples with one element from every
class and original sizes summing to $B$. For $k>3$, the dummy classes preserve
the perfect-matching question: any collection of disjoint valid original triples
can be extended using distinct dummy elements, and any disjoint family of target
$k$-tuples projects to disjoint valid original triples.

Set
\[
H=1,
\qquad
S=100\cdot H\cdot Q^2,
\qquad
M=S\cdot Q.
\]
For each element $e$, create an independent high-or-zero box equal to $M$ with
probability $1-1/b_e$ and zero otherwise. Give an element $e$ in class $j$ cost
\[
c_e=S\cdot\ln b_e=S\cdot(\ln p_j+s(e)\ln2),
\]
and set the setup cost to
\[
T=S\cdot(Q-1-\ln Q)-H
=S\cdot\left(Q-1-\sum_{j=1}^k\ln p_j-B\ln2\right)-H.
\]

For a nonempty batch $A$, write
\[
w(A)=\prod_{e\in A}b_e,
\qquad
K(A)=T+\sum_{e\in A}c_e.
\]
The batch fails with probability $1/w(A)$, and its one-shot surplus is
\begin{align*}
M\cdot\left(1-\frac{1}{w(A)}\right)-K(A)
&=S Q\cdot\left(1-\frac{1}{w(A)}\right)
  -\left(S(Q-1-\ln Q)-H+\sum_{e\in A}S\ln b_e\right)\\
&=S Q-\frac{S Q}{w(A)}
  -S(Q-1-\ln Q)+H-S\sum_{e\in A}\ln b_e\\
&=H+S\left(1+\ln Q-\ln w(A)-\frac{Q}{w(A)}\right)\\
&=H-S\cdot\left(\ln\frac{w(A)}{Q}+\frac{Q}{w(A)}-1\right)\\
&=H-S\cdot D(w(A)/Q).
\end{align*}
Thus every target batch has surplus $H$. If $A$ is non-target, then
\cref{lem:analytic-product-separation} gives
\[
H-S\cdot D(w(A)/Q)
\leq H-\frac{S}{16Q^2}
=-\frac{21}{4}\cdot H,
\]
so every non-target batch is overpriced.

Every target batch has the same item-cost sum
\[
S\cdot\left(\sum_{j=1}^k\ln p_j+B\ln2\right),
\]
and therefore the same total cost
\[
K=S\cdot(Q-1)-H,
\]
the same failure probability $q=1/Q$, and the same positive one-shot surplus
\[
\Delta=M\cdot(1-q)-K=H.
\]

Let $\nu$ be the maximum number of pairwise-disjoint target batches. By
\cref{app:lem:overpriced}, an optimal policy opens only target batches. Until
the first success, every observed value is zero, so the all-failure path of any
deterministic policy is a sequence of disjoint target batches and has length at
most $\nu$. Conversely, any packing of $r$ target batches can be opened in any
order until the first success. Its expected utility is
\begin{align*}
M\cdot(1-q^r)-K\cdot(1+q+\cdots+q^{r-1})
&= \bigl(M\cdot(1-q)-K\bigr)\cdot(1+q+\cdots+q^{r-1})\\
&= \Delta\cdot(1+q+\cdots+q^{r-1}).
\end{align*}
This expression is strictly increasing in $r$. Randomization and stopping
before exhausting a packing cannot improve it, since every additional available
target batch has the same positive conditional surplus. Hence
\[
\operatorname{OPT}_1
= \Delta\cdot(1+q+\cdots+q^{\nu-1}),
\]
with the geometric sum interpreted as zero when $\nu=0$.

A YES instance has $\nu=m$, while a NO instance has $\nu\leq m-1$. Define
\[
\Theta=
\frac{1}{2}\cdot
\left[
\Delta\cdot(1+q+\cdots+q^{m-1})
+
\Delta\cdot(1+q+\cdots+q^{m-2})
\right].
\]
Then $\operatorname{OPT}_1\geq\Theta$ exactly in YES instances.

\emph{Representation.}
The construction above used exact logarithms. To make the input rational, let
$s_{\max}=\max_e s(e)$ and
\[
D_0=2\cdot k+k\cdot s_{\max}+B.
\]
By \cref{lem:rational-log-surrogates}, choose rationals
$\ell_1,\ldots,\ell_k,\ell_{\mathrm{bin}}$ such that
\[
|\ell_j-\ln p_j|\leq\delta,\qquad
|\ell_{\mathrm{bin}}-\ln2|\leq\delta,\qquad
\delta=\frac{H}{100\cdot S\cdot D_0}.
\]
Replace the real costs by
\[
c_e=S\cdot(\ell_j+s(e)\cdot\ell_{\mathrm{bin}})
\]
for an element $e$ in class $j$, and set
\[
T=S\cdot\left(Q-1-\sum_{j=1}^k\ell_j-B\cdot\ell_{\mathrm{bin}}\right)-H.
\]
For any batch of at most $k$ elements, the change in its surplus is at most
\[
S\cdot(2\cdot k+k\cdot s_{\max}+B)\cdot\delta
=H/100.
\]
Target batches still have total cost exactly $K=S(Q-1)-H$ and surplus exactly
$H$, while every non-target batch remains overpriced. The setup cost remains
positive because the real setup cost $S(Q-1-\ln Q)-H>0$ for $Q\geq8$ and the
perturbation is at most $H/100$; the item costs are nonnegative by the same
choice of precision.

Finally, the reduction has polynomial size. Since $B$ and all sizes $s(e)$ are
bounded by a polynomial in $m$, and $k$ is fixed, the integers $2^B$ and
$Q=(\prod_j p_j)2^B$ may be large as values but have polynomial binary encoding
length. Hence the labels $b_e$ also have polynomial binary encoding length. The
probabilities $1-1/b_e$, the scale parameters $S,M$, the rational logarithmic
approximations, and the rational costs all have polynomial encoding length.
Also $Q^{m-1}$ has bit length $O(m\log Q)=O(mB)$, so the threshold $\Theta$ is
polynomially representable. Thus the rational instance and threshold are
produced in polynomial time.
\end{proof}

\subsection{Proof of Theorem~\ref{thm:bp-sign-hard}: Sign-gap hardness for non-reusable boxes}

Before proving the non-reusable sign-gap theorem, we isolate the batch-level
hardness used in the reduction. For a feasible batch $B$, let $X_B$ be the best
value observed inside the batch and let $K(B)$ be its opening cost. Its
\emph{one-shot surplus} is
\[
\Delta(B)=\mathbb{E}[X_B]-K(B).
\]
If $\Delta(B)<0$, equivalently $K(B)>\mathbb{E}[X_B]$, the batch is
\emph{overpriced}.

The construction makes the probability of no success equal to the reciprocal of
a product. Costs are logarithmic in the same labels. Thus the reward term
contributes a reciprocal-product term, the cost term contributes a logarithmic
term, and after calibration the surplus contains $D(w/Q)$, which is uniquely
minimized at the target product $w=Q$.

\begin{definition}[Restricted \textsc{Subset Product}]
The input consists of integers $a_1,\ldots,a_n\geq2$ and an integer target
$Q\geq8$. The question is whether there is a subset $S\subseteq[n]$ such that
\[
\prod_{i\in S}a_i=Q.
\]
\end{definition}

The usual \textsc{Subset Product} problem has no restrictions on $a_i$ or $Q$.
The restricted version above remains NP-hard.

\begin{lemma}[Restricted Subset Product Remains NP-Hard]
\label{lem:restricted-subset-product}
\textnormal{Restricted \textsc{Subset Product}} is NP-hard.
\end{lemma}

\begin{proof}
Reduce from the standard \textsc{Subset Product} problem. Given an instance
$(a_1,\ldots,a_n,Q)$, delete all labels equal to $1$, since they do not affect
feasibility, and write
\[
A=\prod_i a_i
\]
for the product of the remaining labels, with $A=1$ if none remain. Add one new
label
\[
P=8\cdot A
\]
and set the new target to
\[
Q'=P\cdot Q.
\]
All labels in the new instance are at least $2$, and $Q'\geq8$.

If the original instance has a subset of product $Q$, then adding the new label
$P$ gives a subset of product $Q'$. Conversely, any subset of the new instance
with product $Q'$ must contain $P$: without it, the product is at most
$A<Q'$. After removing $P$, the remaining selected labels have product $Q$.
Thus the two instances are equivalent, and the transformation has polynomial
bit length because $\log A$ is at most the total input length.
\end{proof}

\begin{lemma}
\label{lem:product-gap}
For every rational $H>0$, deciding whether a common-prize high-or-zero
non-reusable instance contains a nonempty positive-surplus feasible batch is
NP-hard, even when the batch capacity is part of the input and also when
$k=\infty$, and even with $T>0$.
\end{lemma}

\begin{proof}
We first describe the reduction using exact logarithmic costs and ignore
representation-size issues; the final paragraph replaces the logarithms by
rational approximations and checks that the constructed instance has polynomial
encoding length. We reduce from restricted \textsc{Subset Product}. Let the
source instance be $(b_1,\ldots,b_N,Q)$. Set
\[
\Lambda=100\cdot H\cdot Q^2,
\qquad
M=\Lambda\cdot Q.
\]
For label $b_j$, create a high-or-zero box equal to $M$ with probability
$1-1/b_j$ and zero otherwise. A set $A$ has failure probability
\[
\prod_{j\in A} b_j^{-1}=\frac{1}{w(A)},
\qquad
w(A)=\prod_{j\in A}b_j.
\]
Give label $b_j$ opening cost
\[
c_j=\Lambda\cdot\ln b_j,
\qquad
T=\Lambda\cdot(Q-1-\ln Q)-H.
\]
The one-shot surplus of $A$ is
\begin{align*}
M\cdot\left(1-\frac{1}{w(A)}\right)
-T-\sum_{j\in A} c_j
&=\Lambda Q\cdot\left(1-\frac{1}{w(A)}\right)
  -\left(\Lambda(Q-1-\ln Q)-H+\Lambda\sum_{j\in A}\ln b_j\right)\\
&=H+\Lambda\left(1+\ln Q-\ln w(A)-\frac{Q}{w(A)}\right)\\
&= H-\Lambda\cdot\left(\ln\frac{w(A)}{Q}+\frac{Q}{w(A)}-1\right)\\
&= H-\Lambda\cdot D(w(A)/Q).
\end{align*}
It equals $H$ when $w(A)=Q$. If $w(A)\ne Q$, then
\cref{lem:analytic-product-separation} gives
\[
H-\Lambda\cdot D(w(A)/Q)
\leq H-\frac{\Lambda}{16\cdot Q^2}
= -\frac{21}{4}\cdot H.
\]

\emph{Representation.}
The construction above used exact logarithms. To make the input rational, apply
\cref{lem:rational-log-surrogates} and choose rational numbers
$\ell_j,\ell_Q$ satisfying
\[
|\ell_j-\ln b_j|\leq\delta,
\qquad
|\ell_Q-\ln Q|\leq\delta,
\qquad
\delta=\frac{H}{100\cdot \Lambda\cdot(N+1)}.
\]
Replace the exact logarithmic costs by
\[
c_j=\Lambda\cdot \ell_j,
\qquad
T=\Lambda\cdot(Q-1-\ell_Q)-H.
\]
For every nonempty set $A$, rationalization changes its surplus by at
most
\[
\Lambda\cdot\left(|\ell_Q-\ln Q|+
  \sum_{j\in A}|\ell_j-\ln b_j|\right)
\leq \Lambda\cdot(N+1)\cdot\delta
= H/100.
\]
Therefore a target batch has surplus in
$[99\cdot H/100,101\cdot H/100]$, while every off-target batch has surplus at
most $-21H/4+H/100<-5H$. Thus the constructed instance contains a
positive-surplus batch if and only if the restricted \textsc{Subset Product}
instance is YES.
The precision also keeps every $\ell_j$ positive, since $b_j\geq2$. Since
$Q\geq8$, the quantity $Q-1-\ln Q$ is bounded away from zero, and the
approximation error is less than $H/100$, so $T>0$. The required precision has
polynomial encoding length, so \cref{lem:rational-log-surrogates} makes the
whole construction polynomial time. Set the batch capacity to $k=N$; since there are
only $N$ boxes, the same instance also proves the claim for $k=\infty$.
\end{proof}

\begin{proof}[Proof of Theorem~\ref{thm:bp-sign-hard}]
Use the instances constructed in \cref{lem:product-gap} with $H=1$. In a YES
instance, there is a batch with positive one-shot surplus, so opening it once
gives $\operatorname{OPT}_1>0$. In a NO instance, every nonempty feasible batch
is overpriced. By \cref{app:lem:overpriced}, there is an optimal non-reusable
policy that never opens an overpriced batch. Hence the empty policy is optimal
and $\operatorname{OPT}_1=0$. Thus a polynomial-time solver for the promise
problem would decide the NP-hard batch-existence problem from
\cref{lem:product-gap}.
\end{proof}

\subsection{Proof of Theorem~\ref{thm:rbp-hard}: Sign-gap hardness for reusable boxes}

\begin{lemma}
\label{lem:multiset-product-gap}
For every rational $H>0$, deciding whether a common-prize high-or-zero reusable
instance contains a nonempty positive-surplus feasible batch is NP-hard when
finite $k$ is part of the input, and also for $k=\infty$, even with $T>0$.
\end{lemma}

\begin{proof}
We first describe the reduction using exact logarithmic costs and ignore
representation-size issues; the final paragraph replaces the logarithms by
rational approximations and checks that the constructed instance has polynomial
encoding length. We reduce from restricted \textsc{Subset Product}. Let the
source instance be $(a_1,\ldots,a_n,Q)$. The only additional difficulty in the
reusable model is that a batch may contain repeated copies. We first encode each
original item as a forced binary choice, so that any product hitting the target
uses each source item at most once.

\emph{Encoding the subset choice.}
Choose pairwise distinct private primes $p_1,\ldots,p_n$ that divide neither
$Q$ nor any input number $a_j$.

For each source item $i$, create two labels
\[
b_{i,0}=p_i,
\qquad
b_{i,1}=p_i\cdot a_i,
\]
and set
\[
Q'=Q\cdot\prod_i p_i.
\]
The label $b_{i,0}$ means ``skip item $i$,'' while $b_{i,1}$ means ``take item
$i$.'' Both labels contain the private prime $p_i$, and no other label or
source factor contains $p_i$.

For a reusable multiplicity vector $u$, define
\[
w(u)=\prod_i b_{i,0}^{u_{i,0}}\cdot b_{i,1}^{u_{i,1}}.
\]
The private-prime encoding gives the following equivalence:
\[
w(u)=Q'
\quad\Longleftrightarrow\quad
u_{i,0}+u_{i,1}=1\ \ \forall i
\quad\text{and}\quad
\prod_{i:u_{i,1}=1}a_i=Q.
\]
Indeed, comparing the exponent of each private prime $p_i$ forces
$u_{i,0}+u_{i,1}=1$, and after canceling all private primes the remaining
condition is exactly the source subset-product equation.

\emph{Creating the reusable sign gap.}
Let $L=n$ for the finite-batch-capacity instance, where we set $k=n$. For the
unbounded-batch-capacity instance, set $L=L_0=4Q'+1$.

Apply the product construction to the $2n$ labels $b_{i,0},b_{i,1}$ with target
$Q'$ and multiplicity bound $L$. Set
\[
\Lambda=100\cdot H\cdot (Q')^2,
\qquad
M=\Lambda Q'.
\]
For each label $b_j$, create a reusable high-or-zero type equal to $M$ with
probability $1-1/b_j$ and zero otherwise. Give label $b_j$ opening cost
\[
c_j=\Lambda\ln b_j,
\qquad
T=\Lambda(Q'-1-\ln Q')-H.
\]
For a multiplicity vector $u$ with $\|u\|_1\leq L$, the all-zero probability is
$1/w(u)$, and its one-shot surplus is
\begin{align*}
M\cdot\left(1-\frac{1}{w(u)}\right)-T-\sum_j u_j c_j
&=\Lambda Q'\cdot\left(1-\frac{1}{w(u)}\right)
  -\left(\Lambda(Q'-1-\ln Q')-H+\Lambda\sum_j u_j\ln b_j\right)\\
&=H-\Lambda\cdot\left(\ln\frac{w(u)}{Q'}+\frac{Q'}{w(u)}-1\right)\\
&=H-\Lambda D(w(u)/Q').
\end{align*}
Thus vectors with $w(u)=Q'$ have surplus $H$, while those with $w(u)\ne Q'$
have surplus at most $-21H/4$ by \cref{lem:analytic-product-separation}.

Since every label is at least $2$, every item cost is at least $\Lambda/4$.
In the finite-batch-capacity instance every feasible batch has $\|u\|_1\leq n=L$. In
the unbounded-batch-capacity instance, every batch with $\|u\|_1>L_0$ has cost greater
than
\[
L_0\cdot\frac{\Lambda}{4} = M+\frac{\Lambda}{4},
\]
so its surplus is at most $-\Lambda/4<-5H$ because its expected reward is at
most $M$.

\emph{Representation.}
The construction above used exact logarithms. To make the input rational, apply
\cref{lem:rational-log-surrogates} to the labels and to $Q'$, and choose
rational approximations with
\[
\delta=\min\left\{\frac{H}{100\cdot\Lambda\cdot(L+1)},\frac14\right\}.
\]
Replace each $\ln b_j$ and $\ln Q'$ in the costs and setup cost by its rational
approximation. Rationalizing the logarithms changes the surplus of every vector
$\|u\|_1\leq L$ by at most $H/100$. Hence target vectors have surplus in
$[99\cdot H/100,101\cdot H/100]$, and every non-target vector of size at most
$L$ has surplus at most $-21H/4+H/100<-5H$. Since $b_j\geq2$ and
$\delta\leq1/4$, every item cost remains at least $\Lambda/4$. Therefore, in the
unbounded-batch-capacity instance, every batch with $\|u\|_1>L_0$ still has surplus at
most $-\Lambda/4<-5H$.

Combining this with the private-prime equivalence, the constructed rational
reusable instance has a positive-surplus feasible batch if and only if the
source \textsc{Subset Product} instance is YES. The private primes can be found
in polynomial time: if $L_{\mathrm{in}}$ is the total source bit length, then
$Q\cdot\prod_j a_j$ has at most $L_{\mathrm{in}}$ distinct prime divisors, so
among the first $n+L_{\mathrm{in}}$ primes at least $n$ are available. Thus the
primes, labels, and target $Q'$ have polynomial binary encoding length. In the
unbounded-batch-capacity case, $L_0=4Q'+1$ may be numerically large, but its bit length
is polynomial. The required precision has polynomial encoding length. The exact
setup cost is positive for $Q'\geq8$, and replacing $\ln Q'$ by its rational
approximation changes it by at most $\Lambda\delta\leq H/100$, so the setup cost
remains positive. Hence the rational instance is produced in polynomial time.
\end{proof}

\begin{proof}[Proof of Theorem~\ref{thm:rbp-hard}]
Use the instances constructed in \cref{lem:multiset-product-gap} with $H=1$. In
a YES instance, there is a reusable batch with positive one-shot surplus. For a
common-prize reusable batch with success probability $s$, cost $K$, and prize
$M$, positive surplus means $K<sM$, so its reservation value is
$M-K/s>0$. By \cref{thm:reusable-index}, this gives $\operatorname{OPT}_1>0$.
In a NO instance, every nonempty reusable batch is overpriced. Then
$\sigma^\star=0$, and \cref{thm:reusable-index} says opening nothing is optimal,
so $\operatorname{OPT}_1=0$. Thus a polynomial-time solver for the promise
problem would decide the NP-hard batch-existence problem from
\cref{lem:multiset-product-gap}.
\end{proof}
\section{Proofs of Bicriteria Approximation Results}
\label{sec:approximation-proofs}

\subsection{Proof of Theorem~\ref{thm:general-noreusable}: Non-reusable boxes}
\label{subsec:proof-nonreusable-approx}

The proof follows the outline from \cref{subsec:approx-nonreusable}. We first
identify the finite-support reward envelope, then verify that every adaptive
policy induces a feasible LP point, and finally analyze the rounding of the
opening marginals and the fractional batch count.

For fixed opening marginals $x\in[0,1]^n$, define the atomized one-winner
envelope
\[
B(x)=
\max\left\{
\sum_i\sum_{\ell\in L_i} v_{i\ell}y_{i\ell}:
0\leq y_{i\ell}\leq p_{i\ell}x_i\ \forall i,\ell,
\quad
\sum_i\sum_{\ell\in L_i}y_{i\ell}\leq1
\right\}.
\]
The useful form of this envelope is its threshold representation. As in the
high-or-zero case, the identity is the same fractional-knapsack structure for a
single winner; the only difference is that each box has been split into its
positive atoms.

\begin{proposition}[Threshold Form of the Atomized Envelope]
\label{prop:threshold-form-general-bx}
For every $x\in[0,1]^n$,
\[
B(x)=
\int_0^\infty
\min\left\{1,\sum_i x_i\Pr[V_i>t]\right\}\,dt .
\]
\end{proposition}

\begin{proof}
For each positive atom $a=(i,\ell)$, write
$w_a=v_{i\ell}$ and $\mu_a=p_{i\ell}x_i$. Then
\[
B(x)=\max\left\{
\sum_a w_a y_a:
0\leq y_a\leq \mu_a\ \forall a,
\quad \sum_a y_a\leq1
\right\}.
\]
For any feasible $y$,
\[
\sum_a w_a y_a
=\int_0^\infty \sum_{a:w_a>t}y_a\,dt
\leq
\int_0^\infty
\min\left\{1,\sum_{a:w_a>t}\mu_a\right\}\,dt .
\]
Since $\sum_{a:w_a>t}\mu_a=\sum_i x_i\Pr[V_i>t]$, this gives the upper bound.

For the reverse inequality, order the atoms in nonincreasing value and greedily
fill one unit of winner mass:
\[
y_a=
\min\left\{\mu_a,\left(1-\sum_{b<a}y_b\right)^+\right\}.
\]
For every threshold $t$ that is not an atom value, the atoms with $w_a>t$ form
a prefix of this order after grouping equal values, and the greedy construction
attains
\[
\sum_{a:w_a>t}y_a
=
\min\left\{1,\sum_{a:w_a>t}\mu_a\right\}.
\]
The two sides can differ only at finitely many atom values, which do not affect
the integral.
\end{proof}

With this envelope in place, the next step is to check that the LP is indeed a relaxation of adaptive policies. This is where the nonanticipatory nature of a policy is used: the decision to open box $i$ is made before the realization of $V_i$ is observed.

When upper-bounding a positive-utility policy, it is without loss of
generality to assume that the policy opens its first batch with probability
one. Indeed, if a policy stops immediately with probability \(1-q\) and
otherwise follows a nonempty policy \(\pi'\), then
\[
    \operatorname{val}_1(\pi)
    =
    q\,\operatorname{val}_1(\pi').
\]
Thus, whenever \(\operatorname{val}_1(\pi)>0\), conditioning on the nonempty
branch weakly increases utility.
\begin{lemma}[General LP Upper Bound]
\label{lem:general-lp-upper}
Every adaptive non-reusable policy \(\pi\) that opens its first batch with probability one induces a feasible LP solution with objective value $\mathbb{E}[M_\pi]-\mathbb{E}[C_\pi]$.
\end{lemma}

\begin{proof}
Let $x_i$ be the probability that box $i$ is opened, and let $z$ be the
expected number of opened batches. Since $\pi$ opens a first batch, uses
nonempty batches, and consumes every opened box,
\[
1\leq z\leq n,
\qquad
\sum_i x_i\leq kz.
\]
Moreover,
\[
\mathbb{E}[C_\pi]=Tz+\sum_i c_i x_i.
\]
Break ties deterministically among boxes attaining the final maximum. For each
positive atom, let $y_{i\ell}$ be the probability that box $i$ is the selected
positive winner and realizes value $v_{i\ell}$. The event that $i$ is opened is
determined before $V_i$ is revealed and is independent of $V_i$, so
\[
y_{i\ell}
\leq \Pr[i\text{ is opened and }V_i=v_{i\ell}]
=x_i p_{i\ell}.
\]
At most one positive atom is selected, so $\sum_{i,\ell}y_{i\ell}\leq1$, and
\[
\mathbb{E}[M_\pi]=\sum_i\sum_{\ell\in L_i}v_{i\ell}y_{i\ell}.
\]
Thus the induced point is feasible and has the claimed objective.
\end{proof}

The upper bound above is fractional in the winner variables. To round it, we
compare the envelope with the reward of an actual random set of opened boxes.
For a set $S\subseteq[n]$, let
\[
R(S)=\mathbb{E}\left[\max_{i\in S}V_i\right],
\qquad
R(\varnothing)=0,
\]
and let $F$ be its multilinear extension. The function $R$ is normalized,
monotone, and submodular.

The next lemma is the rank-one correlation-gap comparison between independent
sampling with marginals $x$ and the one-winner envelope $B(x)$.

\begin{lemma}[Correlation Gap for the Winner Envelope]
\label{lem:general-random-set-gap}
For every $x\in[0,1]^n$,
\[
F(x)\geq \left(1-\frac1e\right)B(x).
\]
\end{lemma}

\begin{proof}
If each box $i$ is independently sampled with probability $x_i$, then
\[
F(x)=
\int_0^\infty
\left(1-\prod_i\left(1-x_i\Pr[V_i>t]\right)\right)dt .
\]
For a threshold $t$, put $s_t=\sum_i x_i\Pr[V_i>t]$. Then
\[
1-\prod_i\left(1-x_i\Pr[V_i>t]\right)
\geq 1-e^{-s_t}
\geq \left(1-\frac1e\right)\min\{1,s_t\}.
\]
Integrating and applying \cref{prop:threshold-form-general-bx} proves the
claim.
\end{proof}

The correlation-gap lemma assumes that the marginals $x$ can be used directly.
The algorithm must also convert the fractional batch count $z$ into an integer
number of batches. The following lemma isolates exactly the loss from that
two-point rounding.

\begin{lemma}[Random Batch-Count Rounding]
\label{lem:load-rounding}
Fix $x\in[0,1]^n$ and $z\geq1$. Let $\bar z$ be obtained by rounding $z$ to
$\floor z$ or $\ceil z$ as in \cref{alg:lp-rounding}, and let
$q_i=\min\{1,(\bar z/z)x_i\}$. Then
\[
\mathbb{E}_{\bar z}[F(q)]
\geq
\left(1-\frac1e\right)h(z)B(x),
\qquad
h(z)=\frac{\floor z+(z-\floor z)^2}{z}.
\]
Moreover, $h(z)\geq2(\sqrt2-1)$ for all $z\geq1$.
\end{lemma}

\begin{proof}
Write $r=\floor z$ and $\theta=z-r$. If $\bar z=r$, then $q=(r/z)x$. Since
$R$ is monotone submodular, $F$ is concave along nonnegative rays, so
$F(q)\geq(r/z)F(x)$. If $\bar z=r+1$, then $q\geq x$ coordinatewise, and
monotonicity gives $F(q)\geq F(x)$. Therefore
\[
\mathbb{E}_{\bar z}[F(q)]
\geq
\left((1-\theta)\frac r{r+\theta}+\theta\right)F(x)
=h(z)F(x).
\]
Combining with \cref{lem:general-random-set-gap} gives the first claim.

For the batch-count factor, keep $r\geq1$ fixed and vary $\theta\in[0,1]$:
\[
h(z)=\frac{r+\theta^2}{r+\theta}.
\]
Its minimum on this interval is attained at
$\theta^\star=\sqrt{r^2+r}-r$ and equals
\[
2\left(\sqrt{r^2+r}-r\right)
=\frac{2}{\sqrt{1+1/r}+1}.
\]
This expression increases with $r$, so the global minimum is the case $r=1$,
equal to $2(\sqrt2-1)$.
\end{proof}

After the number of batches is integral, the remaining requirement is the
ex-post capacity bound. Since capacity is uniform and batches can be formed
after the set is chosen, pipage rounding supplies the final conversion from
fractional marginals to a feasible set.

\begin{lemma}[Uniform-Capacity Pipage Rounding]
\label{lem:uniform-rounding}
Let $f:2^{[n]}\to\mathbb{R}_{\geq0}$ be normalized, monotone, and submodular,
with multilinear extension $F$. If $q\in[0,1]^n$ and $\sum_iq_i\leq K$, then
randomized pipage rounding returns $S$ with $|S|\leq K$, $\Pr[i\in S]=q_i$, and
\[
\mathbb{E}[f(S)]\geq F(q).
\]
\end{lemma}

This is the standard pipage-rounding guarantee for monotone submodular
functions under a uniform matroid constraint
\cite{calinescu2011maximizing}.

We now combine the ingredients: the LP upper bound for the benchmark, the
correlation-gap comparison for reward, the batch-count rounding loss, pipage
rounding for hard capacity, and the cost accounting.

\generalnoreusablethm*

\begin{proof}
Let $\Phi$ be the optimum value of the LP. By
\cref{lem:general-lp-upper}, \(\Phi\) upper-bounds every nonempty adaptive
policy at original costs.

If $\Phi\leq0$, every nonempty policy has nonpositive original-cost value, and
the empty policy is optimal. Assume $\Phi>0$, and let $(x,z,y)$ be an optimal
solution. For this fixed $(x,z)$, replacing $y$ by an optimizer in the
definition of $B(x)$ can only increase the objective, so we may write
\[
\Phi=B(x)-Tz-\sum_i c_i x_i.
\]

Apply \cref{alg:lp-rounding}. For each realized $\bar z$,
\[
\sum_i q_i\leq\frac{\bar z}{z}\sum_i x_i\leq k\bar z,
\]
so \cref{lem:uniform-rounding} applies with the rank-$k\bar z$ uniform
matroid. Together with \cref{lem:load-rounding}, the expected reward is at least
\[
\alphaNR B(x),
\qquad
\alphaNR=\left(1-\frac1e\right)2(\sqrt2-1).
\]
The expected setup cost is at most $T\mathbb{E}[\bar z]=Tz$. For every item,
\[
\Pr[i\in S]
=(1-\theta)\frac{\floor z}{z}x_i
+\theta\min\left\{1,\frac{\floor z+1}{z}x_i\right\}
\leq x_i,
\]
so the expected item cost is at most $\sum_i c_i x_i$. Hence
\[
\operatorname{val}_{\alphaNR}(\widehat\pi)
\geq
\alphaNR\left(B(x)-Tz-\sum_i c_i x_i\right)
=\alphaNR\Phi.
\]
Since $\Phi\geq\operatorname{OPT}_1$ when the optimum is nonempty, and the
empty optimum has value zero, the theorem follows. Solving the LP only to
additive accuracy $\varepsilon$ loses the corresponding additive amount.

The LP has polynomially many variables and constraints in the finite-support
input representation. Moreover, the multilinear extension
can be evaluated in polynomial time by summing over the finitely many support thresholds. Hence the pipage-rounding procedure is implementable in polynomial time.
\end{proof}

\subsection{Proof of Theorem~\ref{thm:reusable-general}: Reusable boxes}
\label{subsec:proof-reusable-approx}

The reusable proof follows the same order as the overview in
\cref{subsec:approx-reusable}. First, a threshold certificate is converted into
an actual repeated policy. Then we show that the feasibility LP contains every integral
batch certificate, and finally that a fractional certificate can be rounded to
an integral batch while losing only the discounted-cost factor.

\begin{proof}[Proof of \cref{lem:reusable-threshold-certificate}]
Let $p=\Pr[X_{\vec{a}}>\tau]$. The certificate implies $p>0$. Repeating
$\vec{a}$ until the first threshold crossing uses a geometric number of
openings, so its $\lambda$-discounted value is
\[
\mathbb{E}[X_{\vec{a}}\mid X_{\vec{a}}>\tau]
-\lambda\frac{K(\vec{a})}{p}
=
\tau+
\frac{f_\tau(\vec{a})-\lambda K(\vec{a})}{p}
\geq \tau.
\]
\end{proof}

The threshold certificate is useful only if the relaxation certifies every
integral batch that could serve as the optimal repeated batch. The next lemma
checks this by assigning the selected excess above threshold to the atom that
realizes it.

\begin{lemma}[\(\mathrm{RLP}_\tau\) Relaxes Integral Batches]
\label{lem:reusable-general-upper}
Let \(\vec a\) be a feasible reusable batch with reservation value
\(\sigma(\vec a)\). Then, for every \(\tau\le \sigma(\vec a)\), there exist
variables \(y_{i\ell}\) such that \(x_i=a_i\) and \((x,y)\) is feasible for
\(\mathrm{RLP}_\tau\).
\end{lemma}

\begin{proof}
Set $x_i=a_i$. After one opening of $\vec{a}$, break ties deterministically
among copies whose value is the batch maximum and exceeds $\tau$. If the maximum
is at most $\tau$, select no atom. Otherwise, let the selected winning copy have
type $i$ and atom $v_{i\ell}$, and define $y_{i\ell}$ as the probability of
this event. Then $\sum_{i,\ell}y_{i\ell}\leq1$, and
\[
\sum_i\sum_{\ell\in L_i}(v_{i\ell}-\tau)^+y_{i\ell}
=f_\tau(\vec{a}).
\]
The event counted by $y_{i\ell}$ implies that at least one of the $a_i$ fresh
copies of type $i$ realizes atom $v_{i\ell}$, so the union bound gives
\[
y_{i\ell}\leq a_i p_{i\ell}=p_{i\ell}x_i.
\]
Moreover, since \(\tau\le \sigma(\vec a)\),
\[
    f_\tau(\vec a)
    \ge
    f_{\sigma(\vec a)}(\vec a)
    =
    K(\vec a).
\]
Therefore the certificate constraint of \(\mathrm{RLP}_\tau\) is satisfied,
and \((x,y)\), with \(x_i=a_i\), is feasible for
\(\mathrm{RLP}_\tau\).

\end{proof}

Conversely, a certified fractional solution must be implemented as a true
batch. We fill the $k$ slots independently according to the fractional copy
counts and compare the resulting threshold tails with the LP winner envelope.

\begin{lemma}[Reusable Slot Rounding]
\label{lem:reusable-general-slot-rounding}
Let $(x,y)$ be feasible for \(\mathrm{RLP}_\tau\). Fill $k$ independent slots
by choosing type $i$ with probability $x_i/k$ and a dummy type otherwise. Let
$\vec{A}$ be the resulting random multiplicity vector. Then
\[
\mathbb{E}[f_\tau(\vec{A})]
\geq
\alphaR\sum_i\sum_{\ell\in L_i}(v_{i\ell}-\tau)^+y_{i\ell},
\qquad
\mathbb{E}\!\left[\sum_i c_i A_i\right]
=\sum_i c_i x_i .
\]
\end{lemma}

\begin{proof}
For $t\geq\tau$, set
\[
q_x(t)=\sum_i x_i\Pr[V_i>t].
\]
Since $\sum_i x_i\leq k$, we have $0\leq q_x(t)\leq k$. One rounded slot
exceeds $t$ with probability $q_x(t)/k$, so
\[
\Pr[X_{\vec{A}}>t]
=1-\left(1-\frac{q_x(t)}{k}\right)^k
\geq
\alphaR\min\{1,q_x(t)\},
\qquad
\alphaR=1-\frac1e.
\]
The LP winner mass above the same level is bounded by the same envelope:
\[
\sum_{i,\ell:v_{i\ell}>t}y_{i\ell}
\leq
\min\left\{1,\sum_i x_i\Pr[V_i>t]\right\}.
\]
Therefore, for every $t\geq\tau$,
\[
\Pr[X_{\vec{A}}>t]
\geq
\alphaR\sum_{i,\ell:v_{i\ell}>t}y_{i\ell}.
\]
Integrating over thresholds gives the excess-reward bound. The cost identity is
linear because each slot chooses type $i$ with probability $x_i/k$.
\end{proof}

It remains to combine threshold monotonicity, slot rounding, and the structural
characterization of optimal reusable policies by a single reservation-index
batch.

\reusablegeneralthm*

\begin{proof}
The case $T=0$ is handled exactly by \cref{prop:zero-setup}, so assume
$T>0$. 
Certification is monotone: if \(\tau\) is LP-certified, then every \(\tau'\le \tau\) is also LP-certified, because lowering the threshold increases the coefficients \((v_{i\ell}-\tau)^+\) and makes the certificate constraint easier to satisfy.

Suppose \(\mathrm{RLP}_\tau\) has a feasible solution $(x,y)$. By \cref{lem:reusable-general-slot-rounding},
\[
\mathbb{E}\left[
f_\tau(\vec{A})
-\alphaR\left(T+\sum_i c_iA_i\right)
\right]
\geq
\alphaR\left(
\sum_i\sum_{\ell\in L_i}(v_{i\ell}-\tau)^+y_{i\ell}
-T-\sum_i c_i x_i
\right)
\geq0.
\]
Since the input distributions have finite support, the conditional expectation
of the displayed certificate can be computed exactly after any partial
assignment of slots. The method of conditional expectations returns a
realization $\vec{a}$ satisfying
\[
f_\tau(\vec{a})\geq\alphaR K(\vec{a}).
\]
Because $T>0$, the all-dummy realization has negative certificate value, so the
returned batch is nonempty. By \cref{lem:reusable-threshold-certificate}, the
repeated policy for this batch has $\alphaR$-discounted value at least $\tau$.

If $\operatorname{OPT}_1=0$, the algorithm either returns the empty policy or a
certified repeated batch, both of which have nonnegative
$\alphaR$-discounted value. Suppose $\operatorname{OPT}_1>0$. By
\cref{thm:reusable-index}, there is an optimal reusable batch
$\vec{a}^{\star}$ with reservation value
$\sigma^\star=\operatorname{OPT}_1$. For every
$\tau\leq\sigma^\star$,
\[
f_\tau(\vec{a}^{\star})
\geq
f_{\sigma^\star}(\vec{a}^{\star})
=K(\vec{a}^{\star}).
\]
\cref{lem:reusable-general-upper} therefore implies that every
$\tau\leq\operatorname{OPT}_1$ is LP-certified. Monotonicity of certification
lets binary search return a certified threshold
$\widehat\tau\geq\operatorname{OPT}_1-\varepsilon$. The policy returned for
$\widehat\tau$ satisfies
\[
\operatorname{val}_{\alphaR}(\widehat\pi)
\geq
\widehat\tau
\geq
\operatorname{OPT}_1-\varepsilon. \qedhere
\]
\end{proof}

\paragraph{AI Disclosure.}
OpenAI's ChatGPT (GPT-5.5) was used to support the writing and mathematical development of the paper, including editing, improving the exposition, and formalizing parts of the analysis. In particular, it was used to develop illustrative examples and to extend results from high-or-zero distributions to general finite-support distributions. The authors made all final decisions regarding the mathematical statements and proofs, and independently re-derived, verified, and revised all AI-assisted mathematical content. The authors assume responsibility for all content.

\bibliographystyle{plain}
\bibliography{references}

\appendix
\section{Additional Structural Details}
\label{app:structural-proofs}

\subsection{A lower bound for consecutive batching}
\label{subsec:consecutive-batching-lower-bound}

We sort boxes by their individual reservation values computed with respect to
item costs only: \(\sigma_i\) is defined by
\[
    \mathbb{E}[(V_i-\sigma_i)^+]=c_i.
\]
Fix an ordering of the boxes. A consecutive-batch policy is a policy that only
opens intervals in the current remaining order.

We first give the full calculation for the example from the main text. The setup
cost is \(T=10\), the batch capacity is \(k=2\), and the boxes are ordered as
\(A,D,B\). Box \(A\) has item cost \(c_A=1\), with \(\Pr[V_A=100]=0.1\) and
\(\Pr[V_A=0]=0.9\). Box \(B\) is deterministic with value \(9\) and item cost
\(c_B=1\). Box \(D\) is deterministic with value \(100\) and item cost
\(c_D=91\).

The reservation values are
\[
    \sigma_A=100-\frac{1}{0.1}=90,\qquad
    \sigma_D=100-91=9,\qquad
    \sigma_B=9-1=8.
\]
Thus the reservation-value order is \(A\succ D\succ B\). The utilities of the
consecutive feasible batches are
\[
\begin{array}{c|ccccc}
S & \{A\} & \{D\} & \{B\} & \{A,D\} & \{D,B\} \\ \hline
u(S) & -1 & -1 & -2 & -2 & -2 .
\end{array}
\]
Indeed, \(u(\{A\})=10-10-1=-1\), \(u(\{D\})=100-10-91=-1\),
\(u(\{B\})=9-10-1=-2\), and both pairs containing \(D\) have maximum value
\(100\) and total cost \(102\). On the other hand,
\[
u(\{A,B\})
    =
    \mathbb{E}[\max\{V_A,V_B\}]-10-1-1
    =
    \bigl(0.1\cdot 100+0.9\cdot 9\bigr)-10-1-1
    =
    6.1>0.
\]
So the non-consecutive batch \(\{A,B\}\) is profitable, while every consecutive
batch is unprofitable.

Moreover, every consecutive-batch policy has value at most zero. If such a
policy ever opens \(D\), then its final reward is at most \(100\), while opening
\(D\) alone already costs \(T+c_D=101\). Thus any policy opening \(D\) has
negative value. If the policy never opens \(D\), then it can only open \(A\) and
\(B\) as singleton batches. The expected marginal reward from opening \(A\) is
at most \(10<T+c_A=11\), and the expected marginal reward from opening \(B\) is
at most \(9<T+c_B=11\). Hence opening either singleton cannot increase expected
value, and the empty policy is optimal among consecutive-batch policies.

We now prove the bicriteria lower bound.

\begin{theorem}
\label{thm:consecutive-batching-half-lower-bound}
Fix \(k\ge 2\) and \(\alpha\in(1/2,1]\). There is a finite-support
non-reusable instance with setup cost \(T=1\) and batch capacity \(k\) such that the
reservation-value order is
\[
    A \succ D_1 \succeq \cdots \succeq D_k \succ B,
\]
where $D_1, \dots, D_k$ are identical decoy boxes, and every consecutive-batch policy
\(\pi\) satisfies
\[
    \operatorname{val}_{\alpha}(\pi)\le 0,
\]
while \(\operatorname{OPT}_1>0\). Consequently, no algorithm restricted to
consecutive batches in this order can guarantee an
\((\alpha,\alpha)\)-bicriteria approximation for any \(\alpha>1/2\).
\end{theorem}

\begin{proof}
The construction has two useful boxes and a block of decoy boxes. The useful boxes
\(A\) and \(B\) are profitable only when opened together: \(A\) is a rare large
prize, \(B\) is a moderate sure prize, and their maximum pays for one setup
cost. The decoy boxes only serve to separate \(A\) from \(B\) in the
reservation value order. They are chosen to have reservation values between
those of \(A\) and \(B\), but their total expected reward is arbitrarily small.

Choose \(\beta\) with \(1/2<\beta<\alpha\). This will be the expected value of
each useful box. Then choose \(d>0\) such that \(\beta+d<\alpha\); \(d\) will
upper-bound the total expected value of all decoys. Since \(\beta>1/2\), choose
\(p\in(0,1)\) small enough that
\[
    (2-p)\beta>1.
\]
Finally choose \(\varepsilon>0\) small enough that
\[
    (2-p)\beta-1-2\varepsilon>0
    \qquad\text{and}\qquad
    \varepsilon<\beta.
\]

Next, we create the boxes.
Box \(A\) has item cost \(c_A=\varepsilon\) and
\[
    \Pr[V_A=\beta/p]=p,
    \qquad
    \Pr[V_A=0]=1-p.
\]
Box \(B\) has item cost \(c_B=\varepsilon\) and deterministic value
\(V_B=\beta\). Thus
\[
    \mathbb{E}[V_A]=\mathbb{E}[V_B]=\beta.
\]
Their reservation values are
\[
    \sigma_A=\frac{\beta-\varepsilon}{p},
    \qquad
    \sigma_B=\beta-\varepsilon,
\]
so \(\sigma_A>\sigma_B\).

It remains to insert \(k\) boxes between \(A\) and \(B\) without giving a
consecutive policy useful value. Choose any target reservation value
\(\sigma_D\) with \(\sigma_A>\sigma_D>\sigma_B\). Each decoy is a rare high
value: choose \(G>\max\{\sigma_D,d/k\}\), let
\[
    \Pr[V_D=G]=\frac{d}{kG},
    \qquad
    \Pr[V_D=0]=1-\frac{d}{kG},
\]
and set
\[
    c_D=\frac{d}{kG}(G-\sigma_D).
\]
Thus, each decoy box has expected value \(d/k\). At the same time, its reservation
value is exactly the prescribed value \(\sigma_D\), since
\[
    \mathbb{E}[(V_D-\sigma_D)^+]
    =
    \frac{d}{kG}(G-\sigma_D)
    =
    c_D.
\]
We create independent copies $D_1, \dots, D_k$ of the same box type $D$. Therefore the decreasing reservation value order is
\[
    A \succ D_1 \succeq \cdots \succeq D_k \succ B.
\]

Now consider any consecutive-batch policy \(\pi\). The key point is that every
nonempty batch pays one setup cost, while a consecutive policy can extract at
most one useful box, plus negligible decoy box value, per opened batch. Let \(N\) be
the number of nonempty batches opened, let \(I_A,I_B\) indicate whether \(A,B\)
are opened, and let \(I_{D_j}\) indicate whether the \(j\)-th decoy box is opened.
If both \(A\) and \(B\) are opened, then at least two batches must be opened:
in the initial order, any interval containing both useful boxes also contains
all \(k\) decoys and therefore has size \(k+2\). Hence, pathwise,
\[
    I_A+I_B\le N.
\]
Also, if \(N=0\), no decoy box is opened, while if \(N\ge 1\), the total expected
value of all opened decoy boxes is at most the total expected value \(d\) of all
decoys. Therefore, pathwise,
\[
    \frac{d}{k}\sum_{j=1}^k I_{D_j}\le dN.
\]

Let \(M_\pi\) be the maximum reward observed by \(\pi\). We upper-bound this
maximum by the sum of all opened rewards. Since rewards are nonnegative,
\[
    M_\pi
    \le
    I_A V_A+I_B V_B+\sum_{j=1}^k I_{D_j}V_{D_j}.
\]
The event that a box is opened is determined before that box's value is
observed, and values are independent. Therefore,
\[
\begin{aligned}
    \mathbb{E}[M_\pi]
    &\le
    \beta\,\mathbb{E}[I_A+I_B]
    +
    \frac{d}{k}\sum_{j=1}^k \mathbb{E}[I_{D_j}]  \\
    &\le
    (\beta+d)\mathbb{E}[N].
\end{aligned}
\]
The setup cost alone is \(N\), and item costs are nonnegative, so
\(C_\pi\ge N\). Thus
\[
    \operatorname{val}_{\alpha}(\pi)
    =
    \mathbb{E}[M_\pi]-\alpha\mathbb{E}[C_\pi]
    \le
    (\beta+d-\alpha)\mathbb{E}[N]
    \le 0.
\]

On the other hand, the unrestricted policy can open the non-consecutive batch
\(\{A,B\}\), which is feasible because \(k\ge 2\). Since \(\beta/p>\beta\),
\[
    \mathbb{E}[\max\{V_A,V_B\}]
    =
    p\cdot\frac{\beta}{p}+(1-p)\beta
    =
    (2-p)\beta.
\]
Therefore the original-cost value of opening \(\{A,B\}\) is
\[
    (2-p)\beta-1-2\varepsilon>0.
\]
Hence \(\operatorname{OPT}_1>0\), while every consecutive-batch policy has
\(\alpha\)-discounted value at most zero. This proves the theorem.
\end{proof}

\subsection{Numerical details for the strict adaptivity-gap example}
\label{app:strict-gap-example}

We tuned parameters of the instance used in Example~\ref{ex:strict-gap} with at most two decimal places. The setup cost is \(T=0.44\), the batch capacity is \(k=2\), and the boxes are as follows. Box \(1\) is deterministic and has value $V_1=0.83$ and cost \(c_1=0.40\). Box \(2\) has distribution
\[
    \Pr[V_2=0]=0.48,\qquad
    \Pr[V_2=0.83]=0.43,\qquad
    \Pr[V_2=100]=0.09,
\]
and cost \(c_2=8.50\). Box \(3\) has distribution
\[
    \Pr[V_3=100]=0.05,\qquad
    \Pr[V_3=0]=0.95,
\]
and cost \(c_3=4.16\).

Consider the adaptive policy that first opens box \(2\). If \(V_2=0\), it opens
\(\{1,3\}\). If \(V_2=0.83\), it opens \(\{3\}\). If \(V_2=100\), it stops.
On the low and medium branches, the expected final reward after opening the
second batch is
\[
    0.05\cdot 100+0.95\cdot 0.83=5.7885.
\]
Therefore the expected reward of the adaptive policy is
\[
    \mathbb{E}[M]
    =
    (0.48+0.43)\cdot 5.7885+0.09\cdot 100
    =
    14.267535.
\]
The relevant batch costs are
\[
    K(\{2\})=0.44+8.50=8.94,\qquad
    K(\{1,3\})=0.44+0.40+4.16=5.00,
\]
and
\[
    K(\{3\})=0.44+4.16=4.60.
\]
Hence
\[
    \mathbb{E}[C]
    =
    8.94+0.48\cdot 5.00+0.43\cdot 4.60
    =
    13.318.
\]
Thus the adaptive policy has value
\[
    \mathbb{E}[M]-\mathbb{E}[C]
    =
    0.949535.
\]

We now compute the best fixed-menu value. Since a fixed-menu policy can ignore
available batches, it suffices to enumerate maximal disjoint menus. For each
fixed menu \(\mathcal{F}\), we treat every batch \(B\in\mathcal{F}\) as a
Pandora box with reward \(X_B=\max_{i\in B}V_i\) and cost \(K(B)\), run
Weitzman's rule, and record the resulting fixed-menu utility
\(\operatorname{val}_1(\mathcal{F})\).

For \(k=2\), the maximal disjoint menus and their values are
\[
\begin{array}{c|c}
\mathcal{F} & \operatorname{val}_1(\mathcal{F}) \\ \hline
\{\{1\},\{2\},\{3\}\} & 0.796055 \\
\{\{1,2\},\{3\}\} & 0.794535 \\
\{\{1,3\},\{2\}\} & 0.795380 \\
\{\{2,3\},\{1\}\} & 0.789055 .
\end{array}
\]
Therefore
\[
    \operatorname{OPT}_{\mathrm{fix}}(T,2)=0.796055.
\]
Since
\[
    \frac{0.949535}{0.796055}>1.1928,
\]
this gives a strict adaptivity gap of about \(1.193\).

For \(k\ge 3\), the only additional maximal menu is the single batch
\(\{\{1,2,3\}\}\). Its one-shot value is
\[
    \mathbb{E}[\max\{V_1,V_2,V_3\}]-K(\{1,2,3\})
    =
    14.267535-13.50
    =
    0.767535,
\]
which is still below \(0.796055\). Hence the same rounded instance gives a
strict adaptive-versus-fixed-menu gap for every \(k\ge 2\).

\subsection{Capped-value reduction for non-reusable policies}

For a feasible batch $B$, let $X_B=\max_{i\in B}V_i$. If $K(B)=0$, define
\[
\kappa_B=X_B,
\qquad
\beta_B=0.
\]
If $0<K(B)\leq\mathbb{E}[X_B]$, let $\sigma_B$ solve
\[
K(B)=\mathbb{E}\bigl[(X_B-\sigma_B)^+\bigr].
\]
For such positive-cost batches, define
\[
\kappa_B=\min\{X_B,\sigma_B\},
\qquad
\beta_B=(X_B-\sigma_B)^+.
\]
A positive-cost batch with $K(B)>\mathbb{E}[X_B]$ is called overpriced.

\begin{lemma}[Overpriced Batches Can Be Skipped]\label{app:lem:overpriced}
In every finite non-reusable instance, there is an optimal policy that never
opens an overpriced batch.
\end{lemma}

\begin{proof}
Let $W(U,m)$ be the optimal expected continuation payoff when the unopened boxes
are $U$ and the current best observed value is $m$. Increasing the incumbent
cannot hurt, and increasing it by $d$ can improve the final selected value by at
most $d$. Thus $W(U,\cdot)$ is nondecreasing and $1$-Lipschitz.

Suppose a policy is about to open an overpriced batch $B\subseteq U$. Its
continuation payoff is at most
\begin{align*}
-K(B)+\mathbb{E} W(U\setminus B,\max\{m,X_B\})
&\leq W(U\setminus B,m)-K(B)+\mathbb{E}\bigl[(X_B-m)^+\bigr]\\
&\leq W(U\setminus B,m)-K(B)+\mathbb{E}[X_B]\\
&< W(U\setminus B,m).
\end{align*}
Discarding $B$ without opening it and then following the same continuation is
therefore strictly better than opening $B$.
\end{proof}

\begin{lemma}[Capped-Value Upper Bound for Adaptive Policies]\label{app:lem:cap-reduction}
Let $O_\pi$ be the random family of batches opened by a finite non-reusable
policy $\pi$ that uses no overpriced positive-cost batches. Then
\[
\operatorname{val}_1(\pi)\leq
\mathbb{E}\left[\max\bigl(\{0\}\cup\{\kappa_B:B\in O_\pi\}\bigr)\right].
\]
\end{lemma}

\begin{proof}
Index the possible rounds by $t$. Let $\mathcal{H}_t$ be the history immediately
before round $t$, let $I_t$ indicate that the round is reached, and let $B_t$ be
the batch selected at that history. Both $I_t$ and $B_t$ are
$\mathcal{H}_t$-measurable. Since $B_t$ contains only unopened boxes, its
primitive values are independent of $\mathcal{H}_t$. Therefore, on every
reached history,
\[
\mathbb{E}[\beta_{B_t}\mid\mathcal{H}_t]
= \mathbb{E}\bigl[(X_{B_t}-\sigma_{B_t})^+\mid\mathcal{H}_t\bigr]
=K(B_t),
\]
and hence $\mathbb{E}[I_t\cdot\beta_{B_t}]
=\mathbb{E}[I_t\cdot K(B_t)]$.

Let $A_t$ indicate, with deterministic tie-breaking, that round $t$ supplies the
final maximum. Then $A_t\leq I_t$, and
\begin{align*}
\mathbb{E}[M_\pi]-\mathbb{E}[C_\pi]
&=\mathbb{E}\sum_t A_t\cdot(\kappa_{B_t}+\beta_{B_t})
  -\mathbb{E}\sum_t I_t\cdot K(B_t)\\
&\leq \mathbb{E}\sum_t A_t\cdot\kappa_{B_t}\\
&\leq \mathbb{E}\left[\max\bigl(\{0\}\cup\{\kappa_B:B\in O_\pi\}\bigr)\right].
\end{align*}
All sums are finite because each primitive box can be opened at most once.
\end{proof}

\begin{proposition}[Fixed Menus Are Classical Pandora Instances]
\label{app:prop:fixed-weitzman}
Fix a finite disjoint family $\mathcal{F}$ of feasible batches, each of which
is either zero-cost or non-overpriced. Treat each batch $B\in\mathcal{F}$ as one
box in a finite classical Pandora instance, with value distribution $X_B$ and
opening cost $K(B)$. These induced boxes are independent, so Weitzman's rule is
optimal, and its value is
\[
\mathbb{E}\left[\max\bigl(\{0\}\cup\{\kappa_B:B\in\mathcal{F}\}\bigr)\right].
\]
\end{proposition}

\begin{proof}
Disjointness makes the induced box prizes $(X_B)_{B\in\mathcal{F}}$ independent.
The claim is Weitzman's theorem applied to these boxes
\cite{weitzman1979optimal}.
\end{proof}

\subsection{Proof of the non-reusable adaptivity gap}

\begin{lemma}[Random-Path Rounding]\label{app:lem:random-path}
Let $Y_1,\ldots,Y_N$ be independent latent variables. Each action $e$ has a
support $S(e)\subseteq[N]$, reveals a measurable observation of
$(Y_i)_{i\in S(e)}$, and carries a nonnegative integrable mark determined by
that observation. Assume every feasible action sequence has pairwise-disjoint
supports. For every finite-depth adaptive decision tree, there is a deterministic
feasible path whose expected maximum mark is at least half the adaptive
expected maximum mark.
\end{lemma}

\begin{proof}
It suffices to prove the claim for the randomized path obtained by running the
tree on an independent ghost instance and then probing the resulting path on a
fresh real instance; one deterministic path attains at least the average.

Fix a threshold $\tau\geq0$. Let $a_\tau(\mathcal{T})$ be the probability that
the adaptive tree $\mathcal{T}$ observes a mark at least $\tau$, and let
$r_\tau(\mathcal{T})$ be the same probability for the ghost-path policy. We
prove $r_\tau(\mathcal{T})\geq a_\tau(\mathcal{T})/2$ by induction on the tree
depth.

Let the root action be $e$, with mark $X_e$, observation $O_e$, and continuation
tree $\mathcal{T}_o$ after observation $o$. Put
$p=\mathbb{P}[X_e\geq\tau]$. Then
\[
a_\tau(\mathcal{T})
=p+\mathbb{E}\left[\mathbf{1}_{\{X_e<\tau\}}
  a_\tau(\mathcal{T}_{O_e})\right].
\]
For the ghost-path policy, the real root mark is independent of the ghost
observation $\widetilde O_e$. Moreover, every feasible continuation has support
disjoint from the root support. The induction hypothesis gives
\[
r_\tau(\mathcal{T})
\geq p+\frac{1-p}{2}\cdot\mathbb{E}\left[
  a_\tau(\mathcal{T}_{\widetilde O_e})\right].
\]
Since $\widetilde O_e$ and $O_e$ have the same marginal distribution,
\[
\mathbb{E}\left[a_\tau(\mathcal{T}_{\widetilde O_e})\right]
\geq \mathbb{E}\left[\mathbf{1}_{\{X_e<\tau\}}
  a_\tau(\mathcal{T}_{O_e})\right]
=:B.
\]
Therefore
\[
r_\tau(\mathcal{T})-\frac{1}{2}\cdot a_\tau(\mathcal{T})
\geq \frac{p}{2}\cdot(1-B)\geq0.
\]
Integrating over thresholds using
$\mathbb{E}[Z]=\int_0^\infty\mathbb{P}[Z\geq\tau]\,d\tau$ proves the expected
maximum bound.
\end{proof}

\begin{proof}[Proof of \cref{thm:bp-gap}]
The definitions above already handle zero-cost batches through
$\kappa_B=X_B$ and $\beta_B=0$. By \cref{app:lem:overpriced}, take an optimal
non-reusable policy that never opens an overpriced positive-cost batch.
\cref{app:lem:cap-reduction}
upper-bounds its net value by the expected maximum capped value among the
batches it opens.

Apply \cref{app:lem:random-path} to the adaptive decision tree whose
actions are feasible batches, whose supports are the primitive boxes in the
batch, and whose mark is the capped value $\kappa_B$. Along every realized
non-reusable path these supports are disjoint. Hence some fixed disjoint family
has expected maximum capped value at least half the adaptive capped value.
\cref{app:prop:fixed-weitzman} identifies this quantity with the
value of the corresponding fixed-menu Pandora instance. Thus
\[
\operatorname{OPT}_1^{\BP(T,k)}\leq
2\cdot\operatorname{OPT}_{\mathrm{fix}}(T,k).
\]
\end{proof}

\subsection{Proof of Theorem~\ref{thm:reusable-index}: Reusable reservation index}

\begin{proof}[Proof of \cref{thm:reusable-index}]
For finite $k$, there are only $\binom{n+k}{k}-1$ feasible reusable batches. For
$k=\infty$, every non-overpriced batch $\vec{a}$ satisfies
\[
T+\sum_i a_i\cdot c_i=K(\vec{a})\leq\mathbb{E}[X_{\vec{a}}]\leq V_{\max}.
\]
Since every $c_i>0$, this implies
$a_i\leq (V_{\max}-T)_+/c_i$ for each $i$, so only finitely many non-overpriced
batches remain. Hence $\sigma^\star$ is attained whenever it is positive.

For every feasible reusable batch $\vec{a}$ and every $s\geq\sigma^\star$,
\[
\mathbb{E}\bigl[(X_{\vec{a}}-s)^+\bigr]\leq K(\vec{a}).
\]
For a non-overpriced batch this follows because the excess function decreases
in $s$ and $s\geq\sigma(\vec{a})$. For an overpriced batch it follows from
$\mathbb{E}[(X_{\vec{a}}-s)^+]\leq\mathbb{E}[X_{\vec{a}}]<K(\vec{a})$.

Set $W(m)=\max\{m,\sigma^\star\}$. If $m<\sigma^\star$, opening any batch $\vec{a}$
and then continuing according to $W$ gives at most
\[
-K(\vec{a})+\mathbb{E}W(\max\{m,X_{\vec{a}}\})
=\sigma^\star-K(\vec{a})+\mathbb{E}\bigl[(X_{\vec{a}}-\sigma^\star)^+\bigr]
\leq \sigma^\star=W(m).
\]
If $m\geq\sigma^\star$, the corresponding bound is
\[
-K(\vec{a})+\mathbb{E}\max\{m,X_{\vec{a}}\}
=m-K(\vec{a})+\mathbb{E}\bigl[(X_{\vec{a}}-m)^+\bigr]
\leq m=W(m).
\]
Backward induction bounds every finite-horizon policy by $W(m)$. For an
admissible infinite-horizon policy, truncate after $h$ openings. Bounded
convergence applies to the terminal reward and monotone convergence applies to
the accumulated cost, so the same upper bound survives as $h\to\infty$.

If $\sigma^\star>0$, let $\vec{a}^{\star}$ attain it and put
$p=\mathbb{P}[X_{\vec{a}^{\star}}>\sigma^\star]$. Positive opening cost implies
$p>0$. Repeating $\vec{a}^{\star}$ until the first threshold crossing uses a geometric
number of openings and has value
\[
\mathbb{E}[X_{\vec{a}^{\star}}\mid X_{\vec{a}^{\star}}>\sigma^\star]
-\frac{K(\vec{a}^{\star})}{p}
=\sigma^\star,
\]
by the reservation equation. This attains $W(m)$ whenever $m<\sigma^\star$; for
$m\geq\sigma^\star$, stopping attains $W(m)$.
\end{proof}

\end{document}